\documentclass[12pt, titlepage]{article}
\pdfoutput=1

\makeatletter
\def\ps@headings{%
\def\@oddhead{\mbox{}\scriptsize\rightmark \hfil \thepage}%
\def\@evenhead{\scriptsize\thepage \hfil \leftmark\mbox{}}%
\def\@oddfoot{}%
\def\@evenfoot{}}
\makeatother
\usepackage{amsmath, amssymb,amsthm, soul}
\usepackage{graphicx}
\usepackage{dsfont}
\usepackage[usenames,dvipsnames]{xcolor}
\usepackage{ifsym, wasysym}
\usepackage{verbatim}
\usepackage[tight,footnotesize]{subfigure}
\usepackage{cite}
\usepackage{url}
\usepackage{comment}
\usepackage{todonotes}
\excludecomment{versiona}
\includecomment{versionb}

\newtheorem{theorem}{Theorem}

\newtheorem{lemma}{Lemma}
\newtheorem{cor}{Corollary}

\def\eqd{{\,{\buildrel d \over =}\,}}
\def\deq{{\,{\buildrel \bigtriangleup \over =}\,}}

\def\conv{\otimes}
\def\deconv{\oslash}
\def\eps{\varepsilon}
\def\P{{Pr}}  %%%   appears in many equations  Prob
\def\E{{E}}  %%%   for "expectation value" (in math mode)
\def\minplus{{$(\min, +)\,$}}
\def\S{{\cal S}}
\def\A{{\cal A}}
\def\D{{\cal D}}

\def\M{{\mathcal M}}

\def\X{{\mathcal X}}

\def\mx{{$(\min, \times)$}}

\graphicspath{{./},{./figs/}}

\begin{document}

%\title{Adaptive Scalable Video Streaming In Multi-Hop Wireless Networks }
\title{Performance Analysis of Reliable Video Streaming with Strict Playout Deadline in Multi-Hop Wireless Networks }
\author{Hussein Al-Zubaidy, Viktoria Fodor, Gy\"{o}rgy D\'{a}n, Markus Flierl\\ %~\IEEEmembership{Student Member,~IEEE,}
\\ School of Electrical Engineering, KTH Royal Institute \\ of Technology,  Stockholm, Sweden. \\
\\ E-mail:  \{hzubaidy, vjfodor, gyuri, mflierl\}@kth.se.
        }%

% make the title area
\maketitle

\begin{abstract}
Motivated by emerging vision-based intelligent services, we consider the problem of rate adaptation for high quality and low delay visual information delivery over wireless networks using scalable video coding. Rate adaptation in this setting is inherently challenging due to the interplay between the variability of the wireless channels, the queuing at the network nodes and the frame-based decoding and playback of the video content at the receiver at very short time scales. To address the problem, we propose a low-complexity, model-based rate adaptation algorithm for scalable video streaming systems, building on a novel performance model based on stochastic network calculus. We validate the model using extensive simulations. We show that it allows fast, near optimal rate adaptation for fixed transmission paths, as well as cross-layer optimized routing and video rate adaptation in mesh networks, with less than $10$\% quality degradation compared to the best achievable performance.

%The recent improvement of the transmission rate in emerging wireless networks, from body are and sensor networks to mobile networks, as well as the emergence of cheap and high quality cameras enhanced the ability of these networks to carry time-sensitive traffic for multimedia applications such as high-definition video streaming and video communications. Furthermore, this traffic may traverse multi-hop  path
%that may include one or more wireless links.  Channel variability, due to fading, may affect the performance of such network and as a result the quality of the received video may suffer. Scalable video streaming may be  used to adapt video quality to the achievable transmission rate.
%In this work we provide a model and a methodology to study the effect that the fading channel parameters of the multi-hop wireless network have on the perceived distortion at the receiver and hence the quality of experience (QoE) of the end user.

\end{abstract}
% no keywords
%\vspace{2mm}
%\begin{IEEEkeywords}
%%% keywords here, in the form: keyword \sep keyword
%\textit{scalable video coding; wireless multimedia; multihop  fading channels;  performance analysis; network calculus }
%\end{IEEEkeywords}

%---------

%\todo[inline]{
% Paper structure proposal
%I. Introduction:
%- what do we do:
%1. Analytic model on the playout rate under outage constraint, over multihop wireless
%2. Optimal adaptation of playout deadline and layer selection
%- what is new:
%1. Frame level model, but modeling packet level buffering on the transmission path
%2. multihop
%II. Related work
%III. Preliminaries
%1. Scalable video coding - layering introduced with equation, we mention distortion, no equations
%2. Network calculus for wireless networks
%IV. Playout performance model
%1. The model
%2. the basic numerical results
%- the idealized system with infinite layers, no overhead
%- the system with fixed number of layers and per layer overhead
%V. Model based adaptive playout control
%1. Playout delay minimization under quality constraint (rate, outage)
%2. Quaility maximization (rate) under delay and outage  constraint.
%VI. Discussion
%}

\section{Introduction}

Low cost cameras that are able to capture high quality images, combined with increasing wireless transmission rates, and advances in video coding and visual processing are
enabling a variety of novel, visual information based intelligent services. The services include vision-controlled robotics~\cite{5gppp-1}, automated driving applications~\cite{Gerla,5gppp-1}, and telematic surgery~\cite{Ghodoussi}, and are often safety critical. Their requirements differ significantly from the ones of traditional video content distribution: they require very low latency, high reliability and good video quality for human inspection or for automated visual processing. As an example, use case specifications for eHealth, future factories and automotive~\cite{5gppp-1,5gppp-2}
require tens of milliseconds of network, and hundreds of milliseconds of application level delay limits, with a 99.99\% reliability.
Since the cameras are often hard to access or are mobile, the use of wireless data transmission is inevitable, likely over multiple wireless hops. Multiple wireless hops facilitate the support for multicast or convergecast of the captured streams of images, facilitate handling node mobility, may help to cope with the hostile wireless environment, and could allow low power transmission for battery driven nodes~\cite{Baroffio,Movassaghi,Nishiyama}.

Scalable video coding (SVC) is an important enabling technology to achieve reliable video streaming in dynamic networking environments~\cite{SchwarzTCSVT07,BoyceTCSVT16,RufenachtTIP16}. For scalable video coding,  each video frame, or group of pictures (GOP), is encoded into multiple layers, and rate adaptation, that is, the selection of an appropriate number of layers to be transmitted, makes it possible to adjust the transmission rate, and thus the video quality to the bandwidth available for the transmission. When the bandwidth of the network path deteriorates, less layers should be transmitted so as to reduce the queuing delays at the intermediate nodes, and to safeguard timely delivery of the layers that are transmitted. At the same time, the transmission of too few layers is detrimental as well, as it leads to high distortion at the receiver.

The traditional approach to rate adaptation with SVC is to follow the long term changes of the transmission rate, either based on the buffer occupancy at the receiver or by estimating the transmission rate~\cite{Spiteri16, DeCiccoPV13, LiJSAC14, YinSIG15}. Long term rate adaptation is then combined with buffering at the receiver, so as to even out the short term bandwidth variations.
Nonetheless, these rate adaptation solutions cannot be applied for wireless network applications with strict delay limits, as the short term variability of the wireless channels can not be compensated with in-network and receiver buffering. The rate adaptation problem under strict delay constraints is thus particularly challenging and calls for a novel solution approach.

In this paper we attack the problem by proposing model-based rate adaptation for low-latency video streaming in wireless networks, { by utilizing a stochastic network calculus approach. A significant advantage of this approach is that it provides quantifiable measures of end-to-end quality of service (QoS) as a function of  link quality. These measures can then be translated into useful quality of experience measures (QoE) for the video playout. The QoE measures can then be used for rate adaptation and performance optimization, as well as for the evaluation of new coding schemes.}

We utilize the wireless extensions of stochastic network calculus to be able to capture the transmission of the bit stream over the time varying wireless links and the queuing delays at the intermediate nodes~\cite{Zubaidy_INFOCOM13,AlZubaidyTON}. We extend previous results to take the playout process into account, combining two different time scales, the bit-stream-based transmission and queuing and the video-frame-based decoding and playout.
%{\color{green} A significant advantage of using stochastic network calculus is that this approach provides  quantifiable measures of end-to-end quality of service (QoS) in terms of the underlying multi-hop wireless channel parameters. These measures can be translated into useful quality of experience measures (QoE), e.g., the quality of video playout, and can also be used for performance optimization  as we do in this work.}
We validate the model through extensive simulations, and demonstrate the efficiency of the model-based source rate adaptation and its extension to cross-layer optimized delay sensitive routing.

The rest of the paper is organized as follows. Section~\ref{sec:RelatedWork} discusses recent results on video playout optimization and stochastic network calculus.  Section~\ref{sec:Preliminaries}  presents background regarding the methodology used. Section~\ref{sec:SysDes} describes the considered system.  Section~\ref{sec:ModelDes} presents the model and provides a lower bound on the playout rate under reliability constraint. The model is validated in Section~\ref{sec:Numerical}, and in Section~\ref{sec:Adaptive} we evaluate the efficiency of the model-based rate adaptation, including also transmission path optimization. Section~\ref{sec:conclusions} concludes the paper.

\section{Related Work}\label{sec:RelatedWork}

% Adaptive video coding
% Layered video standards (skip it?)
% Network calculus

%{\color{red}{Markus, do you know references that can help us to justify that i) we consider layered video with many layers (more than 4....), ii) we consider that layers have fixed size for each video frame, that is, for given layers, we have transmit constant rate video stream.}}

%In this paper we model the performance of layered video transmission in wireless networks.
%A prominent example of layered video coding is the Scalable Video Coding (SVC) extension of the  H.264/MPEG-4 advanced video coding (AVC)  video compression standard. The objective of SVC is to provide spatial and temporal scalability with similar decoding complexity and compression performance as single layer coding schemes
%%, which is successfully achieved as it is demonstrated in
%\cite{SchwarzCSVT07}\cite{XinCSVT13}.
%The SVC encoded bit stream is constructed by encoding each video frame into multiple layers, a base layer and several enhancement layers, using hierarchical prediction structure. As a result, an enhancement layer can be decoded if all previous layers are fully received.  The layered structure facilitates a trade-off

{
As the importance of SVC is widely recognized, scalable extensions of video coding standards are available for H.264/AVC \cite{SchwarzTCSVT07}, as well as SHVC for H.265/HEVC \cite{BoyceTCSVT16}, and new, highly scalable solutions are subject to current research \cite{RufenachtTIP16}.
%\cite{Sullivan13, Yan14},
The objective of SVC is to provide temporal, spatial, and quality scalability by encoding the video stream into multiple layers, a base layer and several enhancement layers, using interlayer processing.}
%The latest version of layered video coding is the scalable extension SHVC \cite{Sullivan13, Yan14} of the H.265/HEVC standard. The objective of SHVC in High Efficiency Video Coding (HEVC) is to provide temporal, spatial, and quality scalability with a multi-loop coding framework. The SHVC encoded bit stream is constructed by encoding the video into multiple layers, a base layer and several enhancement layers, using interlayer processing.
As a result, an enhancement layer can be decoded if all previous layers are fully received. The layered structure facilitates a trade-off
between video quality (i.e., distortion) and required bandwidth (i.e., rate).
This property becomes handy for transmission over wireless and mobile networks \cite{Nightingale13}, where the underlying link quality is subject to the channel variability
%due to the random (fading) gain and due to user mobility, as it is demonstrated for several networking and application scenarios in
\cite{SchierlCSVT07,LinCL12,ChenICNSC14,ChenTCSVT15}. The same property makes SVC attractive for delay-limited applications, since all layers received completely before the playout deadline can be utilized for the decoding.
%{\color{green} It is worth noting that encoding with a larger number of layers increases the potential of more efficient rate adaptation. However, the actual SVC implementation used may limit this potential due to complexity and/or efficiency constraints. Nevertheless, the rapid technological advances may soon render such limitations obsolete. In light of the above, one of the advantages of this work is that it provides a methodology -- based on stochastic network calculus -- to quantify the performance gain of SVC with an arbitrarily large number of layers.}
{ Encoding with a larger number of layers increases the potential of more efficient rate adaptation. However, the actual SVC implementation used may limit this potential due to complexity and/or efficiency constraints. Nevertheless, the rapid technological advances may soon render such limitations obsolete. It is therefore important to quantify the performance gains when using a high number of layers in  various application domains.}

Proposed rate adaptation methods for SVC are based on buffer content \cite{Spiteri16}, transmission rate estimation \cite{DeCiccoPV13,LiJSAC14}, or both \cite{YinSIG15}, with the advantage that detailed modeling of the network performance is not required. Low delay applications however can not build on buffer-content-based models. Results presented in the literature consider tens of seconds of playout delays. Similarly for low latency requirements, rate adaptation based on average transmission rate   would be overly optimistic; it would result in queuing delays at the network nodes and late arrivals at the playout buffer. Therefore, in this paper we propose rate adaptation based on network performance modeling   for low latency wireless applications.

Performance modeling of adaptive video streaming in wireless networks has mostly been considered for a single wireless link. In~\cite{LinCL12} the effect of an unreliable wireless channel is modelled by an i.i.d packet loss process, and the video coding rate and the packet size are optimized under retransmission-based error correction.
%taking various transmission overheads into account.
In~\cite{ChenTCSVT15} and~\cite{YangTMM11} adaptive media playout and adaptive layered coding is addressed respectively. Both papers define a queuing model on a video frame level, assuming that the wireless channel results in a Poisson frame arrival process at the receiving terminal, a simplification that may be reasonable if the buffering at the receiver side is significant, and therefore packet level delays do not need to be taken into account.

 Modeling of video streaming based on network calculus  is presented in~\cite{RizkNW15} for the purpose of resource allocation in cellular networks, again, considering a frame level model. Modeling of video transmission over two wireless links is presented in~\cite{SongGLOB11}. This work considers the video transmission as a bitstream, but even with this simplifying assumption the results reflect that modeling based on traditional queuing theory  quickly becomes  intractable as the number of links increases. In \cite{WuMNA05} a tractable model is derived for the delay violation probability for fluid transmission over multihop wireless links, following the effective capacity concept. This approach however does not lend itself to frame level modeling.

In this paper we propose model-based rate adaptation utilizing network calculus.
Network calculus characterizes the departure process and the network backlog over multihop paths. Together with recent advances on modeling wireless links, this motivates our approach.

Stochastic network calculus has been extended to capture the randomly varying channel capacity of wireless links, following different methods
%Recently, various approaches have been proposed to extend the results for networks of wireless links, e.g.,
~\cite{SigmetricsCiucu11,Fidler-Fading,Mahmood_Rizk_Jiang,Verticale:2009, Zubaidy_INFOCOM13,FidlerMGF}.
%, where the single most challenging aspect is the randomly varying channel capacity, due to noise and fading.
%Recently, there have been many attempts to use a similar approach to analyze the performance of  wireless networks, .  These efforts highlighted the difficulties of modeling and analyzing  wireless networks.
%The communication link in wireless networks is prone to noise and fading which results in randomly varying  received SNR, and hence randomly varying channel capacity. This proved to be the single most challenging aspect of developing wireless network calculus.
Most of the existing work
%in this area work around this difficulty by assuming an
builds on an abstracted finite-state Markov channel (FSMC)  model of the underlying fading
channel, e.g.,~\cite{Fidler-Fading,Mahmood_Rizk_Jiang} or uses moment generating function based network calculus~\cite{FidlerMGF}.
%Then applying results from the \minplus-based, MGF network calculus developed by Fidler~\cite{FidlerMGF} to obtain performance bounds for wireless networks.
However, the complexity of the resulting models limits the applicability of these approaches in multi--hop  wireless network analysis with more than a few state FSMC model and more than two hops.
%  Flow transformation  due to loss, dynamic routing or retransmissions was studied in \cite{florin11} by defining a (virtual) scaling element to compensate for the lost/added traffic to the original flow. This technique is based on an earlier work \cite{Fidler-Scaling} that provided such scaling element in deterministic settings.
%
%In \minplus~network calculus, a performance analysis is usually proceeds by characterizing  the cumulative  service processes for the network elements as given by Eqs. (\ref{eq:service})--(\ref{eq:service2}). However, as pointed out by \cite{Zubaidy_INFOCOM13}, working with the logarithm of random fading processes may be intractable and would require approximations or abstraction of the fading model, such as finite-state Markov channel model.
%
In this work, we follow the approach proposed by Al-Zubaidy et al~\cite{Zubaidy_INFOCOM13}, where a wireless network calculus based on the \mx~dioid algebra was developed.
%We build our analysis on the  \mx~network calculus, developed for the analysis of multi--hop fading channels.
The main premise for this approach is that the channel capacity, and hence the offered service of fading channels is related to the instantaneous received SNR through the logarithmic function as expressed by the Shannon capacity,   $C(\gamma) = \log(1+\gamma)$. Hence, an equivalent representation of the channel capacity in an isomorphic transform domain, obtained using the exponential function, would be $e^{C(\gamma)} = 1+\gamma $. This simplifies  the otherwise cumbersome computations of the  end-to-end performance metrics.

%The next section provides a brief description of the used approach, that is network calculus, and its application to wireless networks performance analysis.

\section{Network Calculus for Wireless Networks} \label{sec:Preliminaries}

Network calculus has been developed to provide an efficient analytic tool for evaluating the quality of service provided by networks with multi-hop transmission path, including the effect of correlated buffering at the network nodes. In network calculus, the generated network traffic at node $k$ in time interval $[\tau,t)$ is characterized by the cumulative arrivals, that is, the real--valued non--negative bivariate process $A_k(\tau,t)$, while the transmission capabilities of node $k$ are described by the process of cumulative services  $S_k(\tau,t)$. The  resulting departure process, $D_k(\tau,t)$,  characterizes the cumulative traffic leaving node $k$.
These processes are non-decreasing in $t$ with $A_k(t,t)= S_k(t,t) = D_k(t,t)=0$ and $A_k(0,t)\ge D_k(0,t)$ for all $t$. The objective of stochastic network calculus is to derive the departure processes for complex network topologies, and based on that express the network performance, typically in terms of probabilistic bounds on the end-to-end delay  $W(t)$, and the backlog $B(t)=A(0,t)-D(0,t)$, characterizing the amount of traffic delayed in the transmission queues of the network.
Network calculus can be used to analyze networks with either packetized or fluid flow traffic and for discrete or continuous time scale. In this work, we consider fluid flow traffic and discrete (slotted) time.

Introduced in~\cite{Zubaidy_INFOCOM13}, the
%\mx~network calculus  considers fluid--flow traffic that is infinitely divisible, and operates in discrete--time domain, where time slots are denoted by $t \in \{0,1,\ldots \}$ and slot duration $\Delta t =1$ time unit.
%Let the cumulative arrivals, service and departures  to a node $k$ during the time interval $[\tau,t)$ be denoted by the  real--valued non--negative bivariate processes $A_k(\tau,t), S_k(\tau,t)$ and $D_k(\tau,t)$ respectively. These processes are non-decreasing in $t$ with $A_k(t,t)= S_k(t,t) = D_k(t,t)=0$ and $A_k(0,t)\le D_k(0,t)$ for all $t$.
%Different from the typical \minplus~network calculus.
\mx~network calculus  transforms the problem into an alternative domain, called SNR domain, where the SNR service process ($\S_i$) is obtained by taking the exponent of the original service process\footnote{We use the calligraphic upper--case letters to represent traffic and service processes in the SNR domain and to distinguish them from their  bit domain (where traffic and service are measured in bits) counterparts.}, i.e., $\S_i = e^{S_i}$. Therefore, we refer to a network element $i$ as \textit{dynamic SNR server}, if it offers a service  $\S_i$ that satisfies the  input--output inequality  \cite{CSChang}, $ \D(0,t) \ge   \A \conv \S_{i} (0,t)$,
%\begin{align}%\label{eq:dynamicserver}
%       \D(0,t) \ge   \A \conv \S_{i} (0,t) \, ,
%\end{align}
%The service element that satisfies the above inequality is referred to as \textit{dynamic SNR server}.
 where the \mx~convolution and deconvolution are respectively defined for any two SNR processes $\X_1(\tau,t)$ and  $\X_2(\tau,t)$ as
\begin{equation*} %\label{eq:convolution}
\X_1 \conv \X_2 (\tau,t) \deq
\inf_{\tau\leq u \le t} \big \{ \X_1 (\tau,u) \cdot \X_2 (u,t) \big \} \, ,
\end{equation*}
% \vspace{-6mm}
\begin{align*} %\label{eq:deconvolution}
\X_1 \deconv \X_2 (\tau,t)
\deq \sup_{ u \leq \tau} \Big \{ \frac{\X_1
(u, t)}{\X_2(u,\tau)} \Big \} \,. \hspace{+15mm} &
\end{align*}

%\todo[inline]{VF: What is the best place of this text? mx~network calculus  considers fluid--flow traffic that is infinitely divisible, and operates in discrete--time domain, where time slots are denoted by $t \in \{0,1,\ldots \}$ and slot duration $\Delta t =1$ time unit.}

%Now that we have full description of the model in the SNR domain, we can use \mx~network calculus to analyze the network and obtain stochastic end--to--end performance bounds.
%Next we briefly introduce the \mx~network calculus and its main results.
%The total network backlog  in the SNR domain at any time $t\ge 0$ is given by
%\begin{equation}
%\B(t) = e^{B(t)} = {\A(0,t) \over \D(0,t)}
%\end{equation}
%and the end--to--end delay is given by
%\begin{equation}
%\W(t) = W(t) = \inf \left \{w \ge 0: {\A(0,t) \over \D(0,t+w)} \le 1 \right \} \, .
%\end{equation}

{ {The key result of network calculus is the possibility to substitute the sequence of service processes on a multi-hop transmission path with a single network service process, $S_{\rm net}$, by concatenating the service processes for all nodes along a path \cite{Jiang-Book}. In the SNR domain}}
 \begin{equation}\label{eq:concatenationN}
\S_{\rm net}(\tau,t) =   \S_1  \conv \S_2  \conv \cdots \conv \S_N (\tau,t) \, .
\end{equation}
In addition, network performance bounds, e.g., end-to-end delay and backlog, can be obtained in terms of the \mx~deconvolution of the SNR arrival and service processes \cite{Zubaidy_INFOCOM13}.

The computation of the \mx~convolution and deconvolution operations are  not straight forward as it involves the evaluation of products and quotients of random processes. Thus, an exact solution for  (\ref{eq:concatenationN}) may not be feasible. Instead, we may use yet another transform, the Mellin transform, to find bounds on these two operations.
The Mellin transform, see \cite{Davies}, is defined for a   nonnegative random variable $Z$ as  $\M_Z(s)= \E[Z^{s-1}]$, for any complex valued~$s$ given that the expectation exists. Then, the following holds \cite{Zubaidy_INFOCOM13}:

\begin{lemma}\label{lem:MTConvDecon}
Let $\S_1(\tau,t)$ and
$\S_2(\tau,t)$ be two independent SNR service processes.
The Mellin transform of $\S_1 \conv\S_2(\tau,t)$, for all $s<1$, is bounded by
 \begin{equation} \label{eq:convMT}
 \M_{ \S_1 \conv  \S_2} (s, \tau, t ) \le \sum_{u=\tau}^{ t}
\M_{ \S_1} (s, \tau,u ) \cdot \M_{\S_2} (s, u, t )\, .
\end{equation}
 The Mellin transform of $\S_1 \deconv\S_2(\tau,t)$, for $s>1$, is given by
  \begin{equation} \label{eq:deconvMT}
 \M_{ \S_1 \deconv  \S_2} (s, \tau, t ) \le \sum_{u=0}^{\tau} \M_{ \S_1} (s, u,t ) \cdot \M_{\S_2} (2-s, u, \tau )\, .
\end{equation}
%and $\M_Z(s)= \E[Z^{s-1}]$ is the Mellin transform  of the nonnegative random variable $Z$, for any complex valued~$s$ and when the expectation exists \cite{Davies}.
\end{lemma}

Lemma \ref{lem:MTConvDecon} above suggests that the Mellin transform of the \mx~convolution/deconvolution of two independent  processes is bounded by a function of their Mellin transforms. { {In the case of wireless networks, the independence follows from the assumption on independent fading on the consecutive wireless links.}} Consequently, network performance bounds can be obtained in terms of the Mellin transforms of the SNR arrival and service processes of that network.

\section{System Model and Problem Formulation}\label{sec:SysDes}
In this section we describe our model of the wireless network and of video streaming,
formulate the rate adaptation { and routing problem, and
provide the corresponding arrival and service models.
  
\subsection{Wireless network model}
We consider a time slotted multi-hop wireless network with a time slot duration of $\Delta t$. We use $t$ to refer to a time slot.
For each wireless link we consider a block fading channel~\cite{McEliece} with Rayleigh fading distribution, with coherence time larger than $\Delta t$. As our focus is not on channel coding, we assume that a channel coding scheme is available at each node, such that each channel provides a service that is equivalent to its instantaneous Shannon channel capacity, $C(\gamma_{k,t})= W\log_2(1+\gamma_{k,t})$ bits/s, where $W$ is the channel bandwidth,
$\gamma_{k,t}$ is the instantaneous SNR at the receiver of channel $k$ at time slot $t$, and we consider that $\gamma_{k,t} \eqd \gamma_{k}, \forall t$, where $\eqd $ denotes equal in distribution, with average $\bar{\gamma}_{k}$. We allow the average SNR $\bar{\gamma}_{k}$ to change over time, but we make the reasonable assumption that it is known. Feedback of the channel state information (CSI) over a single link is implemented or can be implemented in modern networks~\cite{Halperin,LTE-book}, while the CSI or estimated SNR values can be collected in a mesh network with the help of the routing protocol~\cite{Accettura}.

%We consider wireless links used exclusively by the considered video stream, e.g., a dedicated WSN link or reserved resource blocks in LTE.
%However, the extension to links that are shared by cross traffic is straightforward using the concept of leftover dynamic server introduced in~\cite{AlZubaidyTON}, as we will discuss in Sec.~\ref{sec:Interf}.

We consider that video has to be streamed between two nodes in the wireless network, as shown in Fig.~\ref{fig:model}.
We refer to a sequence of links from the sender to the receiver node as a transmission path, and denote it by $\mathcal{P}$. Furthermore, we denote by $N$ the length of the path.
There may be multiple paths between the sender and the receiver.
We assume that buffers at intermediate nodes are locally FIFO, i.e., frames and their contents are served according to the order of their arrival.
Furthermore, no dropping or duplication of contents is allowed at these nodes.

\subsection{Scalable video streaming}
%As shown on Fig.\ref{fig:model}, we consider video transmission over multihop wireless links, and a fixed playout deadline scheme, where a frame $i$, generated at $\tau_i$ is played out with a playout delay $T_D$, that is, at time $\tau_i+T_D$.

%{\color{red}{We consider video transmission over multihop wireless links (as shown in Fig.\ref{fig:model}) and a fixed playout deadline scheme, where a frame $i$ that is generated at time $\tau_i$ is played out with a playout delay $T_D$, i.e.,  at time $\tau_i+T_D$.
%We consider a time slotted system, where time slots, denoted by $t \in \{0,1,\ldots \}$, have a duration of $\Delta t$.
%
%The captured video is composed of $n$ frames  per second. The generated frames are fed directly to SVC encoder that generates $L \leq L_{\max}$ layers of size $m$ bits  and a header of $h$ bits per layer, resulting in a constant rate traffic of $R_E= nr=(m+ h)nL$ bits/second, where $r$ is the frame size in bits, and is determined by the number of transmitted layers $L$.
%
%The objective of the transmitting terminal is to select $L$ as well as the transmission path, based on feedback on the channel quality of the links, such that the video distortion at the receiving terminal is minimized. The optimal operating point is obtained  by avoiding network congestion on one hand and underutilization on the other.

{  The video is captured at a rate of $n$ frames per second. Depending on the considered coding scheme, a frame can represent a single image, or a group of pictures (GOP). The captured frames are fed directly to the SVC encoder that generates $L \leq L_{\max}$ layers.
  We consider that each layer has a size of $m$ bits and a header of $h$ bits, resulting in a constant rate traffic of $R_E= nr =(m+ h)nL$ bits per second,
  where $r$ is the frame size in bits, and is determined by the number of transmitted layers $L$.
  The assumption of equal size layers and constant rate traffic is desirable for ease of presentation, but the proposed methodology can handle any type of traffic,
  including variable rate video, as long as its Mellin transform exists and is attainable.}

  The coded video frames are transmitted over a wireless network of $N$ transmission links. Once transmitted over the wireless network, received bits are stored in a playout buffer. Frames are decoded and played out regularly with $T_f=1/n$ time intervals and a fixed playout delay $T_D$, that is, a frame $i$ that is generated at time $\tau_i$ is played out after a fixed delay $T_D$ at time $\tau_i+T_D$. According to the layered coding only the completely received layers are used for decoding. Due to the variability of the wireless channels, the number of layers of a frame $i$ that are received within a deadline $T_D$ is random, leading to a varying per frame playout bitrate $R_D$ at the decoder, and consequently varying distortion.}

%That is, we are looking for the maximum target playout rate -- and the corresponding transmission rate -- that can be ensured with high probability, given by a reliability constraint.

%The proposed rate adaptation is based on the mathematical model of the streaming performance, and on information on the link qualities, that is, the SNR values and the fading distribution.

\begin{figure}
\centering
\includegraphics [width=3.8in]{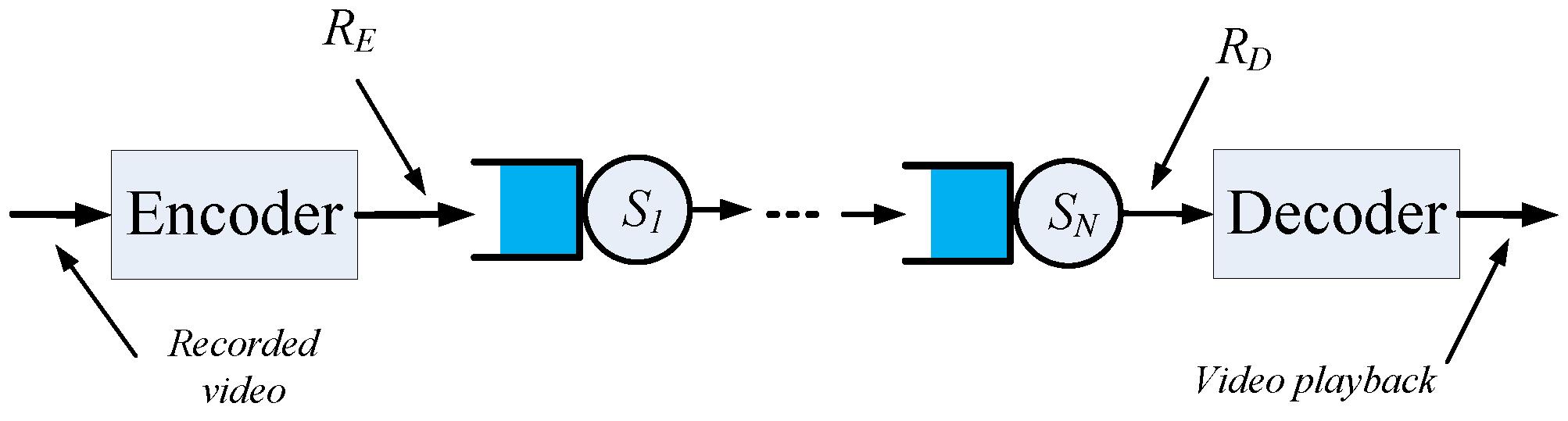}
%\vspace{-4mm}
\caption{Video transmission over multi-hop wireless network.}
%\vspace{-4mm}
\label{fig:model}
\end{figure}

\subsection{Model-based rate adaptation and routing problem}
Given the system model presented above, 
the objective of the model-based rate adaptation { and routing} problem is to select the optimal number of transmitted layers $L^*$
and the optimal transmission path $\mathcal{P}^*$ that together maximize the lower bound $r_D^\varepsilon$ of the playout rate under a reliability constraint $\varepsilon$:
\begin{eqnarray}
\max_{({L},{\mathcal{P}})} r_D^\varepsilon  \label{eq::objective}\\
\textrm{s.t.} \;\;\;\;\;\;\;\;\;\quad\quad\quad\quad\quad\quad && \nonumber\\
\P(R_D < r_D^\varepsilon) &\le& \varepsilon  \label{eq::constraint}
\end{eqnarray}

Due to the variability of the wireless channels and the queuing at the intermediate nodes, there is no tractable analytic expression
for the distribution of $R_D$. Therefore, in the following we provide a bound on the tail distribution of $R_D$, as a function of
the frame size, the average SNR values and the path length $N$. We then show how to use it for solving the model-based rate adaptation and routing problem.

\subsection{Arrival and Service  Models}

Recall that $R_E$ is the bitrate of the coded video,
and is thus the arrival rate to the wireless network. We can then express the cumulative arrival process as
\begin{equation}\label{eq:arrival}
A(\tau,t) = R_E (t- \tau) = (m+ h)nL(t- \tau) \, ,
\end{equation}
and the SNR arrival process   $\A$ is given by
\begin{equation}\label{eq:SNR-Arrival}
\A(\tau,t) = e^{ R_E (t- \tau)} = e^{ (m+ h)nL(t- \tau) } \,   .
\end{equation}
Hence, the Mellin transform of the arrival process can be expressed as
\begin{equation}\label{eq:MA}
\M_{\A}(s,\tau,t) =  e^{(s-1) (m+ h)nL(t-\tau) } \,   .
\end{equation}
%{ \color{red} \st{Note that, although $R_E$ in this model is deterministic, the methodology can still be applied to time-varying random traffic sources as long as their Mellin transform is attainable.}}

%The video frames are transmitted over a wireless network of $N$ transmission links. For each link we consider a block fading channel~\cite{McEliece} with Rayleigh fading distribution.
%As our focus is not on channel coding, we assume that a channel coding scheme is available, such that the channel provides a service that is equivalent to the instantaneous Shannon channel capacity,
%$C(\gamma_{k,t})= W\log(1+\gamma_{k,t})$ bits/s/Hz, where $W$ is the channel bandwidth,
%$\gamma_{k,t}$ is the instantaneous SNR  at the receiver of channel $k$ at time slot $t$. In this work, we assume that $\gamma_{k,t} \eqd \gamma_{t}, \forall k$, where $\eqd $ denotes equal in distribution. Nevertheless, the scheme may also work for non-identically distributed channel fading using a more complex representation of the network service process as shown in \cite{Petreska_2015}.

Similarily, we can define the cumulative service process of a fading channel with 
SNR $\gamma_{k,u}$ is
\begin{equation}\label{eq:service}
S(\tau,t)= W \sum_{u=\tau}^{t-1} \log (1+ \gamma_{k,u}) \, ,
\end{equation}

Its SNR domain counterpart is given by the log-free form
\begin{equation}\label{eq:SNR_service1}
\S(\tau,t) = \prod_{u=\tau}^{t-1} (1+\gamma_{k,u})^{ W}\, .
\end{equation}
The Mellin transform of $\S$ depends on the distribution of $\gamma_{k,u}$, i.e., the fading distribution. In Section \ref{sec:Rayleigh}, we will derive the Mellin transform of the service process for Rayleigh channels.

\subsection{Service Model with Interfering Flows}\label{sec:Interf}
 The service model presented above for a single flow can be extended to capacity sharing between flows. Let us denote by $\A_o$ the SNR arrival process of the tagged (through) flow, and by $\A_c$ that of the other (cross) flows. We can then describe the service offered to the through flow by the leftover service process, which can be characterized as shown in Lemma \ref{lem:leftover} taken from \cite{AlZubaidyTON}.

\begin{lemma} \label{lem:leftover}
Consider a network with a through flow $\A_o$ and cross traffic flow $\A_c$. Assume that the network provides a dynamic SNR server to the aggregate of the two flows, with service process	  $\mathcal S(\tau,t)$ then
\begin{align*}
\mathcal S_o(\tau,t)
= \max \left \{1, \frac{\mathcal S(\tau, t)}{\mathcal A_c(\tau,t)} \right \}
\end{align*}
is a dynamic SNR server satisfying for all $t \geq 0$  that
\begin{align*}
\mathcal D_o(0,t)\ge \mathcal A_o\otimes \mathcal S_o(0,t)
\end{align*}
\end{lemma}
The proof of Lemma 2 can be found in \cite{AlZubaidyTON}. In the rest of the paper, for notational simplicity, we consider that there is no cross traffic.

%}

\section{Performance of Video Communication} \label{sec:ModelDes}
In this section we present our main contribution: A system model for adaptive video transmission over a multi-hop wireless network and a bound on the received video quality  in terms of the parameters of the transmitted video, as well as the underlying fading channels' parameters. 
We first derive a general expression on the lower bound of the received rate under playout delay constraint and frame based transmission. Then we give the bound for transmissions over multihop wireless channels, specifically considering Rayleigh fading. Finally we derive the bound on the playout bitrate, considering the layered structure, and address the feasibility of solving the optimal rate adaptation problem.

\subsection{Lower Bound for the Received Rate } \label{sec:main}

%In order to study the distortion of the received images, we need to estimate the bit rate at the decoder. Then using the provided mapping for the specific decoder, one can compute the distortion. For the SVC under investigation, this means that we need to estimate the number of layers received completely by the decoder within the playback deadline $T_D$.

%To support our calculations,
We investigate a video decoder which operates  as follows: At time $\tau_i+T_D$ it considers the content of the playout buffer. It then drops all content that belongs to  frames $j<i$ (i.e., late arrivals from previous frames), then removes and decodes all frame content that belongs to frame $i$; arrivals from subsequent frames remain in the playout buffer. The modelling challenge is twofold. First, the received rate should include only data that belongs to a given frame. Second, we would like to derive a lower bound of the received rate $R^i_D$, while network calculus usually considers its upper bound to characterize backlog and delay.

A statistical description of the rate at the decoder, $R^i_D$, can be obtained by observing the departure process of the wireless network, $D(\tau, t)$. Specifically, $R^i_D$ can be obtained  by considering all departures during the time period from frame generation until playout, that is, $D(\tau_i, \tau_i +T_D)$, and then counting only $D^i(\tau_i, \tau_i + T_D)$, the part of the traffic that belongs to frame $i$. Since the instantaneous received rate at the decoder, $R^i_D$, includes  only traffic  that belongs to fully received layers by the frame's playout deadline, we can write
\begin{equation}\label{eq:RD_di}
R^i_D = { \left \lfloor{ D^i(\tau_i, \tau_i + T_D) \over m+h } \right \rfloor \cdot m \over T_f}
\, .
\end{equation}

A probabilistic lower bound on the departures belonging to frame $i$ during the period $[\tau_i, \tau_i + T_D)$,  $D^i(\tau_i, \tau_i + T_D), i=1,2, \cdots$, is given by the following lemma.

\begin{lemma} \label{lem:Di}
Given a frame $i$
%arriving at the network in Fig.~\ref{fig:model}
generated at time $\tau_i$ and destined to a decoder with playout deadline $T_D$, the departure process $D^i(\tau_i, \tau_i + T_D), i=1,2, \cdots$,  is characterized as follows
\begin{align}\label{eq:Di_bound}
 & \P (D^i(\tau_i, \tau_i +T_D) \le d) \le \eps \,  \notag \\
& \Leftrightarrow  \P (D(0, \tau_i +T_D)   \le d + (m+h)nL \tau_i )   \le \eps  \, ,
\end{align}
for all $d \le (m+h)L$ and all $\eps \in [0,1]$.
\end{lemma}

It is worth noting that this probability is equal to 1 for all $ d > (m+h)L$, i.e., departures belonging to a frame $i$ can never exceed the frame size.

\begin{figure}
\centering
\includegraphics [width=4.2in]{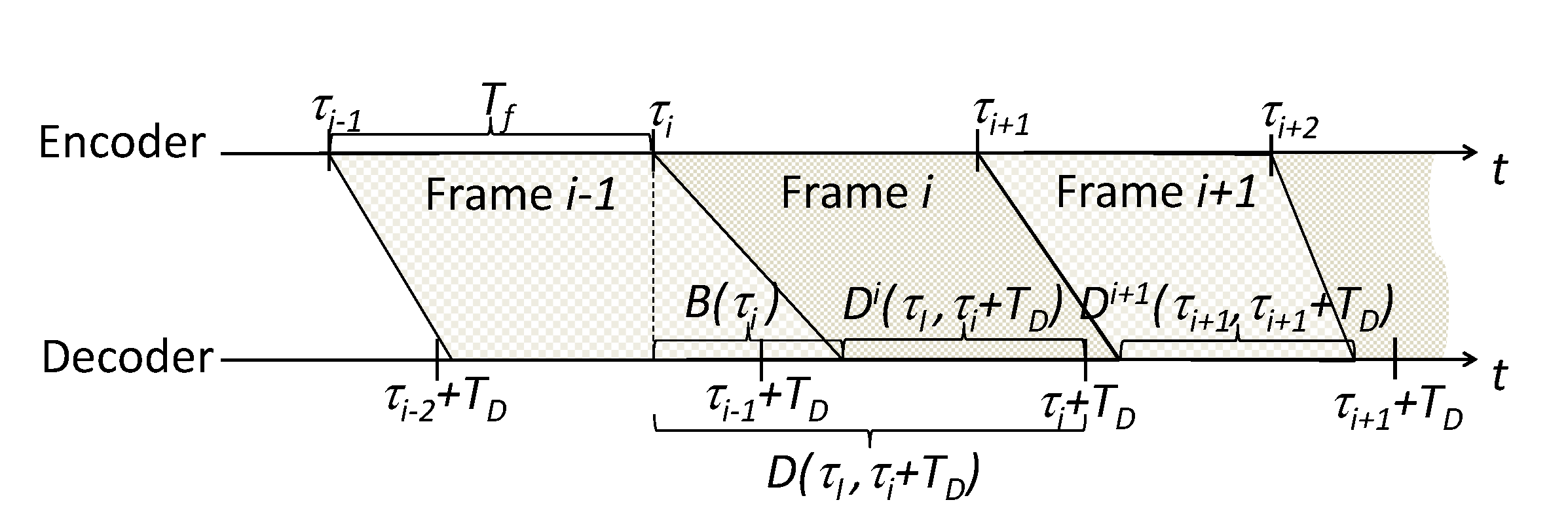}
%\vspace{-4mm}
\caption{Determination of $D^i(\tau_i, \tau_i + T_D)$.}
%\vspace{-3mm}
\label{fig:playoutmodel}
\end{figure}

\begin{proof}
Fig. \ref{fig:playoutmodel} shows the encoding, transmission, and decoding of the consecutive $i-1,i,i+1$ frames and can be used to derive $D^i(\tau_i, \tau_i + T_D)$.
%Departures from the wireless network within $[ \tau_i, \tau_i +T_D)$ may contain traffic that belongs to frames before frame $i$, as well as data of frames after frame $i$. Therefore,
We express $D^i(\tau_i, \tau_i + T_D)$ by first considering all departures $D(\tau_i, \tau_i +T_D)$ and then removing traffic that does not belong to frame $i$. As shown in Fig. \ref{fig:playoutmodel}, the departures belonging to previous frames within the interval $[ \tau_i, \tau_i +T_D)$ are equal to the backlog $B(\tau_i)$ at time $\tau_i$, i.e., the traffic from all previous frames that is still in the network when the $i^{th}$ frame arrives. Once we remove $B(\tau_i)$ the remaining departures belong to frame $i$, up to a size of $(m+h)L$, followed by traffic from subsequent frames. Using this argument we arrive at the following equivalence statement
%
%
%%First we compute all departures (from the wireless network) within $[ \tau_i, \tau_i +T_D)$, then  removing those departures that belong to previous frames, i.e., all traffic generated before $\tau_i$, and finally,  removing departing traffic  that belongs to subsequent frames (see Fig. \ref{fig:playoutmodel}). The latter is done by limiting frame $i$'s departures to a maximum size of $(m+h)L$.
%%
%
%%We are interested in the departures during the period $[ \tau_i, \tau_i +T_D) $, where frame $i$ starts transmitting at $\tau_i$, that belongs to the $i^{th}$ frame, i.e., excluding traffic from the previous frames that is still in the network, denoted by $D^i(\tau_i, \tau_i + T_D)$.
%
%%The departures belonging to previous frames within the interval $[ \tau_i, \tau_i +T_D)$, are equal to the backlog at time $\tau_i$, i.e., the traffic from all previous frames that is still in the network when the $i^{th}$ frame arrives. Let us denote the backlog at  time $\tau_i$ as $B(\tau_i)$.
%
%%Then, we can express the probabilistic lower bound on the departures $D^i(\tau_i, \tau_i + T_D), i=1,2, \cdots$ as
%
%%
%%To obtain a lower bound on the distortion (per frame) we need to estimate an upper bound on the rate $R_D$ during the period $[ \tau_i, \tau_i +T_D)$. This is achieved by computing a probabilistic upper  bound on the departures $D^i(\tau_i, \tau_i + T_D), i=1,2, \cdots$.
%%In other words, for $T_D \le {1 \over n} $ we are looking for
%
%
\begin{align}\label{eq:Di}
 & \P (D^i(\tau_i, \tau_i +T_D) \le d) \le \eps \, \notag \\
& \Leftrightarrow \P (D(\tau_i, \tau_i +T_D) - B(\tau_i) \le d)  \le \eps  \, ,
\end{align}
for all $d \le (m+h)L$.
%We can further simplify \eqref{eq:Di} by considering

Then using the fact that the backlog at any time $\tau$ is given by the difference of all arrivals and all departures from time $t=0$, where $B(0) =0$, until time $t=\tau$, the right hand side of \eqref{eq:Di} can be  evaluated as follows
%
% Therefore the above equation can be evaluated as follows
\begin{align}\label{eq:Di_bound1}
 \P (& D(\tau_i, \tau_i +T_D) - B(\tau_i) \le d)  \notag\\
 & =  \P (D(\tau_i, \tau_i +T_D) - A(0, \tau_i) + D(0, \tau_i) \le d)  \notag \\
 & =  \P (D(0, \tau_i +T_D) - A(0, \tau_i)  \le d) \notag  \\
 & =  \P (D(0, \tau_i +T_D)   \le d + A(0, \tau_i))  \notag \\
 & =  \P (D(0, \tau_i +T_D)   \le d + (m+h)nL \tau_i ) \, .
\end{align}
Substituting \eqref{eq:Di_bound1} in \eqref{eq:Di}, the lemma follows.
\end{proof}

Lemma \ref{lem:Di} states that a probabilistic lower bound for $D^i$ can be obtained in terms of the probabilistic lower bound on the departure process $D$  given by  \eqref{eq:Di_bound1}, for the specific arrival process described in  \eqref{eq:arrival}.
% and \eqref{eq:service} respectively.
When the arrival and service processes have stationary increments, which is the case here, then  $D^i$ is identically distributed for all $i$.
The next step is to derive this bound, which we can accomplish using    network calculus.

\begin{lemma} \label{lem:LB}
For any work-conserving server with dynamic bivariate service process $S(\tau,t)$
and an arrival process $A(\tau,t)$, the departure process $D(\tau,t)$ is bounded as

\begin{equation} \label{eq:LowerBoundD}
D (\tau,t) \ge A \oplus S (\tau,t)=
\inf_{\tau \leq u \le t} \big \{ A (\tau,u) + S (u,t) \big \}
 \, ,
\end{equation}
where $\oplus$ denotes the \minplus~convolution.
\end{lemma}

\begin{proof}
Using Reich's recursive backlog formula we have
\begin{align*} %\label{eq:convolution}
B (t) &= \left [ B(t-1) + a(t) -s(t) \right]^+ \\
& \le \sup_{0 \leq u \le t} \big \{ A (u,t) - S (u,t) \big \}
 \, ,
\end{align*}
where $a(t), s(t), B(t)$ are the instantaneous arrival, service and backlog at time slot $t$ respectively.
Hence,
\begin{align*} %\label{eq:convolution}
D (0,t) &= A(0,t) - B(t) \\
& \ge \inf_{0 \leq u \le t} \big \{ A (0,u) + S (u,t) \big \}
& = A \oplus S (0,t)
 \, ,
\end{align*}
where in the second step we used the fact that $A (0,u)=A (0,t) -A (u,t)$.

But due to causality we also have
$$D (0,\tau) \le A (0,\tau) \, .$$
Then for work-conserving server and for $0 \le \tau \le t$ we get
\begin{align*} %\label{eq:convolution}
D (\tau,t) &= D(0,t) - D(0, \tau) \\
& \ge A \oplus S (0,t) - A(0, \tau) \\
& =\inf_{0 \leq u \le t} \big \{ A (0,u) -A(0,\tau) + S (u,t) \big \}
 \, .
\end{align*}
Since $A (0,u) -A(0,\tau) = 0 \, , \,\forall u<\tau $, and since
$$\exists u \ge \tau \quad \text{s.t.} \quad A(\tau, u) + S (u,t) \le  S(\tau, t)   \, .$$
Then,
$ %\label{eq:convolution}
D (\tau,t) \ge \inf\limits_{\tau \leq u \le t} \big \{ A (\tau,u) + S (u,t) \big \} $
% \, ,
%\end{equation*}
 and the lemma follows.
\end{proof}

\subsection{Departure Process Lower Bound for Wireless Channels}

To use Lemma \ref{lem:LB}, we must evaluate the right hand side of \eqref{eq:LowerBoundD} which is not an easy task for a wireless channel where $S(\tau,t)$ is a randomly varying process due to random fading. Therefore, with the following theorem we provide a probabilistic bound on $D(\tau,t)$ in terms of the Mellin transform of the arrival and service process by using the \mx~network calculus approach \cite{Zubaidy_INFOCOM13}.
% and \cite{AlZubaidyTON}.

\begin{theorem} \label{thm:LowerBound}
Let $\A$ be the SNR arrival process to a work-conserving queuing system with SNR service process $\S$,
%given by \eqref{eq:SNR_service1}
then for any $0 \le \tau \le t$ and any $s<1$,   a lower bound $d$ ($d \ge 0$) for the departure process $D(\tau,t)$ must satisfy the following inequality
\begin{align} \label{eq:DBound}
\P (D (\tau,t) \le d) &  \le  e^{(1-s)d} \sum_{u = \tau}^t \M_{\A} (s,\tau,u) \cdot \M_{\S} (s,u,t) \, .
\end{align}
\end{theorem}

\begin{proof}
 We start by formulating a probabilistic lower bound on the departures  in terms of the SNR departure process $\D$, $ \forall s<1 $, as follows
\begin{align} \label{eq:DBound1}
\P (D (\tau,t) \le d) & = \P(\D(\tau,t) \le e^d) \notag\\
&  = \P(\D^{s-1}(\tau,t) \ge e^{(s-1)d})\,   \notag\\
&  \le e^{(1-s)d} \M_{\D } (s,\tau,t)
%\notag \\
%&\hspace{-10mm} \le  e^{(1-s)d} \sum_{u = \tau}^t \M_{\A} (s,\tau,u) \cdot \M_{\S} (s,u,t)
\, ,
\end{align}
where we used the assumption $s<1 $ to obtain the second line and then we applied Markov's inequality and used the definition of the Mellin transform to arrive at the last step.
%In order to use the approach proposed in \cite{Zubaidy_INFOCOM13} for wireless networks performance analysis we need to obtain a bound on the SNR departure process, $\D$.

Using Lemma \ref{lem:LB}, the SNR departure process $\D$ can be bounded as follows
\begin{align} \label{eq:convolution}
\D (\tau,t) =e^{D (\tau,t)} &\ge e^{\inf_{\tau \leq u \le t} \big \{ A (\tau,u) + S (u,t) \big \}  } \notag  \\
&= \inf_{\tau \leq u \le t} \big \{ e^{A (\tau,u) + S (u,t)} \big \} \notag \\
& =  \inf_{\tau \leq u \le t} \big \{ \A (\tau,u) \cdot \S (u,t) \big \} \notag \\
& = \A \otimes \S (\tau,t)  \, ,
\end{align}
where we used the definition of the \mx~convolution in the last step.

Note that when $s >1$, the Mellin transform  is order-preserving. On the other hand, when $s <1$, the order is reversed \cite{Zubaidy_INFOCOM13}.
Hence, the Mellin transform for the SNR departure process   for any $s<1$ is computed using \eqref{eq:convolution} as follows
\begin{align} \label{eq:convolution1}
\M_{\D} (s,\tau,t) & \le \M_{\A \otimes \S } (s,\tau,t) \notag \\
& =\E \left[ \left( \inf_{\tau \leq u \le t} \big \{ \A (\tau,u) \cdot \S (u,t) \big \} \right)^{s-1} \right ] \notag \\
& = \E \left[  \sup_{\tau \leq u \le t} \Bigr \{ \left( \A (\tau,u) \cdot \S (u,t) \right)^{s-1} \Bigr \}  \right ] \notag \\
& \le \sum_{u = \tau}^t \M_{\A} (s,\tau,u) \cdot \M_{\S} (s,u,t) \, ,
\end{align}
where we used the non-negativity of $\A$ and $\S$ and their independence and then we applied the union bound in the last step.

Substituting \eqref{eq:convolution1} into \eqref{eq:DBound1} the theorem follows.
%
%%A bound on $\M_{\D}$ can then be obtained in terms of a lower bound on the service process and  an exact characterization of the arrival process.
%
%%A lower bound on the arrival process may not be sufficient to obtain a lower bound on the departures since such a bound will result in an overestimation of the network performance which may relax the bound.
%%, i.e., the obtained bound will be a lower bound for an upper bound which is not useful to our application.
%%
%
%
%%
\end{proof}

%Given the  arrival and service processes of our system. The Mellin transform of the considered arrival process given by \eqref{eq:MA}. Unfortunately, an exact expression for

\subsection{ Lower Bound   for Multihop Rayleigh Channels} \label{sec:Rayleigh}
We will now use Theorem~\ref{thm:LowerBound} to obtain a probabilistic lower bound on the departure process for an $N$-hop wireless network subject to Rayleigh fading.
{For simplicity,  we present results for the case when for  $\bar{\gamma}_k=\bar{\gamma}$ for the $N$ hops,
  but the methodology works for non-identically distributed channel
  fading using a more complex representation of the network service process as shown in \cite{Petreska_2015}.}

The instantaneous SNR $\gamma_t$ of a Rayleigh fading channel is exponentially distributed with average $\bar{\gamma}$. Then the Mellin transform for the cumulative service process of a Rayleigh fading channel defined in \eqref{eq:SNR_service1} is  given by \cite{Zubaidy_INFOCOM13}
\begin{align} \label{eq:Rayleigh_service}
\M_{\S} (s,\tau,t) & \le \left( e^{1 \over \bar{\gamma}} \bar{\gamma}^{s-1} \Gamma(s, { \bar{\gamma}^{-1}}) \right)^{t-\tau} \, ,
\end{align}
where $\Gamma(s, a) = \int_a^\infty x^{s-1} e^{-x} dx $ is the incomplete Gamma function.

\begin{theorem}\label{thm:RayleighNet-LB}
A probabilistic lower bound on the departure of $N$-hop i.i.d. Rayleigh channels with average SNR $\bar \gamma$, when the arrival process is given by \eqref{eq:arrival},  and for $0 \le \tau \le t$ is
\begin{align} \label{eq:DBoundRayleigh1}
  \P (D (\tau,t)\!  \le \!  d(t \! -\!  \tau) ) \!
 \le \! \inf_{s<1}\! \left\{\!  e^{(s-1)((m+ h)nL(t-\tau)-d (t-\tau)) } \over \left(1- V(1-s)  \right)^{N}\!\! \!  \right \}
\end{align}
whenever the stability condition
\begin{equation}\label{eq:StabilityCondition}
V(1-s) \deq e^{(1-s) (m+ h)nL} e^{1 \over \bar{\gamma}} \bar{\gamma}^{s-1} \Gamma(s, { 1 \over \bar{\gamma}} ) <1
\end{equation}
is satisfied.
\end{theorem}

\begin{proof}
Let the   service  offered by a network of $N$ store and forward nodes be characterized by the SNR service  process $S(s,\tau,t)={\S_{\rm net}} (s,\tau,t)$.
%\subsection{Multi-hop Rayleigh channel}
%The instantaneous SNR $\gamma_t$ of a Rayleigh fading channel is exponentially distributed with %average $\bar{\gamma}$. The Mellin transform for the cumulative service process of a Rayleigh fading channel is then given by
%\begin{align*} %\label{eq:convolution}
%\M_{\S} (s,\tau,t) & \le \left( e^{1 \over \bar{\gamma}} \bar{\gamma}^{s-1} \Gamma(s, {1 \over %\bar{\gamma}}) \right)^{t-\tau} \, ,
%\end{align*}
%where $\Gamma(s, a) = \int_a^\infty x^{s-1} e^{-x} dx $ is the incomplete Gamma function.
Then using Theorem \ref{thm:LowerBound} we obtain for all $s<1$
\begin{align} \label{eq:DBound-Net}
\P (D (\tau,t) < d) &  \le  e^{(1-s)d} \sum_{u = \tau}^t \M_{\A} (s,\tau,u) \cdot \M_{\S_{\rm net}} (s,u,t) \, .
\end{align}

A bound on $\M_{\S_{\rm net}} (s,\tau,t)$ for $N$  i.i.d. Rayleigh fading channels is obtained by using the server concatenation property \eqref{eq:concatenationN} and then  applying the convolution bound in Lemma \ref{lem:MTConvDecon}  repeatedly $N-1$ times.
Then substituting  \eqref{eq:Rayleigh_service} we have for all $s <1$
% based on Lemma \ref{lem:MTConvDecon}, and the per link service process in   \eqref{eq:Rayleigh_service}, is given by
%the service curve for the $N$-hop wireless network with i.i.d. Rayleigh fading channels based on Lemma \ref{lem:MTConvDecon}, and the per link service process given by \eqref{eq:Rayleigh_service} as follows.
%
\begin{align} \label{eq:NetService}
\M_{\S_{\rm net}} (s,\tau,t) & \le \binom{N-1+t-\tau}{t-\tau} \M_{\S } (s,\tau,t)\notag \\
&\hspace{-10mm} \le \binom{N-1+t-\tau}{t-\tau}  \left( e^{1 \over \bar{\gamma}} \bar{\gamma}^{s-1} \Gamma(s, { 1 \over \bar{\gamma}} ) \right)^{t-\tau}   \, .
\end{align}
The binomial coefficient is the result of expanding the $N-1$ sums and then collecting all terms for the i.i.d channels case \cite{AlZubaidyTON}.

%The Mellin transform for the arrival process is given by
%\begin{equation}\label{eq:MA}
%\M_{\A}(s,\tau,t) =  e^{(s-1) (m+ h)nL(t-\tau) } \,   .
%\end{equation}

Substituting  \eqref{eq:arrival} and \eqref{eq:NetService} in \eqref{eq:DBound-Net}
we obtain the following probabilistic lower bound
\begin{align} \label{eq:DBoundRayleigh}
& \P (D (\tau,t) \le d )
 \le  e^{(1-s)d } \sum_{u = \tau}^t e^{(s-1) (m+ h)nL(u-\tau) } \notag \\
& \hspace{2.5cm} \cdot \binom{N-1+t-u}{t-u}  \left( e^{1 \over \bar{\gamma}} \bar{\gamma}^{s-1} \Gamma(s, { 1 \over \bar{\gamma}} ) \right)^{t-u} \notag \\
&  \qquad\le  e^{(s\!-\!1) ((m\!+\! h)nL(t-\tau)-d) }
 \sum_{v = 0}^\infty  \binom{N\!-\!1\!+\!v}{v}  \left( V(1\!-\!s) \right)^{v}
 \, ,
\end{align}
where, we use the change of variables $v= t-u$, let $t \rightarrow \infty$ and define $V(1-s) \deq e^{(1-s) (m+ h)nL} e^{1 \over \bar{\gamma}} \bar{\gamma}^{s-1} \Gamma(s, { 1 \over \bar{\gamma}} )$.

Using the binomial identity
$$ \sum_{v=0}^\infty \binom{N-1+v}{v} x^v = {1 \over (1-x)^N} \, ,$$
for all $N \ge 1$ and $|x| <1$, the sum in \eqref{eq:DBoundRayleigh}  converges to the following
\begin{align*} %\label{eq:DBoundRayleigh12}
  \P (D (\tau,t) \le d(t- \tau) )
 \le { e^{(s-1)((m+ h)nL(t-\tau)-d (t-\tau)) } \over \left(1- V(1-s) \right)^{N} }
 \, ,
\end{align*}
for all $s<1$, whenever the condition $V(1-s) <1$ is satisfied.
Optimizing over $s$ results in the best possible bound and concludes the proof.
\end{proof}

\subsection{A Bound on Playout Bitrate $R_D$}
{
Combining the  results obtained in Lemma~\ref{lem:Di} and  Theorem~\ref{thm:RayleighNet-LB}  for stable system operation we can compute a lower bound on the departures  $D^i(\tau, \tau +T_D) $ for all $s<1$ as follows
\begin{align} \label{eq:DBoundRayleigh11}
  \P (D^i(\tau, \tau +T_D) &\le d ) =
   \P (D(0, \tau +T_D) \le d+ A(0, \tau)) \notag \\
   & =  \P (D(0, \tau +T_D) \le d+ (m+h)nL \tau  )  \notag \\
 &\le \inf_{s<1} \left \{ e^{(s-1)((m+ h)nL T_D - d ) } \over \left(1- V(1-s) \right)^{N} \right \}
 \, ,
\end{align}
if $d  \le (m+h) L  $, otherwise, i.e., if $d  > (m+h) L  $, $\P (D^i(\tau, \tau +T_D)\le d) =1$.

To obtain the lower bound on the departures  such that $\P (D^i(\tau,t) \le d^{\varepsilon} ) \le \varepsilon$,  we equate the right hand side of (\ref{eq:DBoundRayleigh11}) to $\varepsilon$ and solve for $d^{\varepsilon}$ to get
\begin{align} \label{eq:DBoundRayleigh2}
 d^{\varepsilon} (T_D)
 &\ge \min \Bigg [ (m+h)L, \, \sup_{s<1} \Big \{ (m+ h)nL T_D   \notag \\
& + {1  \over 1-s }  \left[ N \log(1- V(1-s) ) + \log \varepsilon \right] \Big \} \Bigg ]
  \,.
\end{align}

Using \eqref{eq:DBoundRayleigh2},
the distribution of the number of   usable bits (i.e., bits received within the frame's playback deadline, $T_D$)  per second is bounded by

%the distribution of the bitrate received by the decoder within any period equal to $T_D$ is bounded by
$$\P(R_D < r_D^{\varepsilon}) \le \varepsilon \, ,$$
where
\begin{equation}\label{eq:RD_NHop}
r_D^{\varepsilon} \ge { \left \lfloor{d^{\varepsilon}(T_D)  \over m+h } \right \rfloor \cdot m \over T_f}
\, .
\end{equation}
For steady state operation this corresponds to the decodable rate per frame.
%Finally, the distortion is mapped from the experimental rate-distortion curve for the given rate.

%*** VF: Hussein, what we write below is somewhat strange. If the encoder sends (m+h)L information, how can the received information be more? I think we have to check what we want to say here. Maybe skip the entire paragraph? *** {\color{blue} Theoretically, $d^{\varepsilon}(T_D)$ defined by the previous Eq (26) can be  greater than $(m+h)L$. I changed (26) and subsequently (27) to avoid this. Read below I excluded $>$, now its $d^{\varepsilon}(T_D)= (m+h)L$. }

Note that the right hand side of (\ref{eq:RD_NHop}) reduces to ${ m L \over T_f}$ when $d^{\varepsilon}(T_D)= (m+h)L$, i.e.,  all layers of the frame are received within the playout deadline $T_D$. This  can happen when  the underlying wireless links have high channel quality during the frame transmission, and it represents the best distortion performance that can be achieved for the given coding scheme.

\subsection{Effect of $T_D$ on Received Video Quality}
The allowable playout delay $T_D$ has a noticeable effect on the received video quality, as stated in the following corollary.
\begin{cor}\label{cor:1}
The lower bound on the  per frame departures   $d^{\eps}(T_D)$ increases linearly in the playout deadline $T_D$, independently from the network and channel conditions, and of the violation probability requirement.
\end{cor}
\begin{proof}
Rewriting    \eqref{eq:DBoundRayleigh2}  as follows
\begin{align} \label{eq:cor1}
 d^{\varepsilon} (T_D)
 &\ge (m+ h)nL T_D   \notag \\
& +  \sup_{s<1} \Big \{ {1  \over 1-s }  \left[ N \log(1- V(1-s) ) + \log \varepsilon \right] \Big \}
  \,,
\end{align}
the corollary follows since the second term of   \eqref{eq:cor1} does not depend on $T_D$.
\end{proof}

}

%Then
%\begin{align} \label{eq:DBoundRayleigh1U}
%  \P (D (\tau,t) \ge \overline{y}(t- \tau) )
% \le { e^{s((m+ h)nL(t-\tau)-\overline{y}(t-\tau)) } \over \left(1- V(s) \right)^{N} }
% \, .
%\end{align}

\subsection{Optimal rate adaptation and routing}

The bounds (\ref{eq:DBoundRayleigh11}) -- (\ref{eq:RD_NHop}) characterize the effects of the system parameters on the overall system performance, under the considered Rayleigh fading process with given $\bar{\gamma}$. Based on these results, we now show how to solve the optimal rate adaptation and routing problem  \eqref{eq::objective}-\eqref{eq::constraint}.

%This in turn enables performing system optimization and adaptation according to \eqref{eq::objective}-\eqref{eq::constraint}. 

{
The selection of the optimal number of layers per frame $L^*$ and transmission path $\mathcal{P}^*$ requires the evaluation of the right hand side (RHS) of \eqref{eq:DBoundRayleigh2} for all possible paths $\mathcal{P}$. As $V(1-s), \, s <1, $ is convex in $L$, the number of layers per frame,
%and hence in $r$,
the RHS of  \eqref{eq:DBoundRayleigh2} can be shown to be concave in $L$.
It can also be easily shown that $V(1-s), \, s <1, $ is convex in  $s$ whenever $V(1-s) < 1$  (see \cite{Petreska-Arxiv}) and hence, $N \log (1- V(1-s))$ in the RHS of \eqref{eq:DBoundRayleigh2} is concave in $s$.
Therefore, for each path $\mathcal{P}$ the optimal $L^*$  and the corresponding bound $d^{\varepsilon}(T_D)$ can be obtained via a  binary search.}

%
%{\color{blue} VF: Hussein, please check the convex/concave properties for V and s}

%Rate adaptation and routing requires the knowledge of the $\bar{\gamma}$ values at the encoder. Feedback of channel state information (CSI) is implemented or can be implemented in modern networks \cite{Halperin,LTE-book}, while the CSI or estimated SNR values can be collected over the multihop path with the help of the routing protocol \cite{Accettura}.
%{\color{blue}
%The change in the channel quality may be caused by a sudden change in the propagation environment, change of transmit power for some nodes due to power constraints or battery depletion, change of  background interference, etc.
%%On the other hand, a change of the service offered to the  video flow under investigation may be the result of having to share the channel with other (incidental) flows in the network.
%The proposed adaptation can still work in the latter case by inserting the leftover service process, $\S_o$, from Lemma \ref{lem:leftover}  instead of $\S$ in \eqref{eq:DBound} and compute \eqref{eq:DBoundRayleigh2} accordingly.
%}

\section{Model validation and performance evaluation}\label{sec:Numerical}

The analytic model described in Section \ref{sec:ModelDes} provides a lower bound on the per frame departures $d$ within the playout deadline $T_D$. Therefore, we first validate the bounds via simulation. Then, we evaluate the effect of the network and video streaming parameters on the received quality, based on the results in \eqref{eq:DBoundRayleigh11} and \eqref{eq:DBoundRayleigh2}. The effect of $T_D$ have been addressed by Corollary \ref{cor:1}.

We consider an SVC scheme that encodes group of pictures (GOP), that is, a frame in the analytic model represents a GOP. One GOP consists of 10 video images. 25 images are generated per second, which results in $n=2.5$ frames per second. The video is coded with $n=4$ to $24$ layers of size $m=100$ kbits of video payload each, resulting in a per frame payload of $r=0.4-2.4$ Mbits, and a video transmission rate of $1-6$ Mbps. For simplicity, we consider $h=0$, since the typical header size is much smaller than the size of a layer. The playout deadline is $T_D=450$~msec, which corresponds to a strict delay constraint for real-time machine-to-machine video delivery. We consider transmission paths of $N=1,3, 5$ links, a channel of bandwidth $W=2.2$~MHz and average SNR of the fading channels in the range of $\bar{\gamma}=6 - 10$. This corresponds to average channel capacities of $C_{avg}= 4.24 - 6.39$~Mbps. We choose a slot duration of $\Delta t = 10$ msec.

We ran the simulations for a period of $10^{10}$ time slots, which allows empirical evaluation of the system performance up to a violation probability of $\eps =10^{-8}$.

%This time scale is typical for physical layer transmission control for many modern wireless systems. We consider a frame transmission rate $n=25$ frames per second, with $4$ to $24$ layers of size $m=100$ kbits of video payload each, resulting in a per frame payload of $0.4-2.4$ Mbits. For simplicity, we consider $h=0$, since the typical header size is much smaller that the size of a layer. %The channel bandwidth considered is $W=22$~MHz.
%The playout deadline is $T_D=45$~msec, which corresponds to a strict delay constraint for real-time machine-to-machine video delivery. We consider transmission paths of $N=1,3, 5$ links and a channel of bandwidth $W=22$~MHz. For the average SNR of the fading channels, $\bar{\gamma}=6 - 10$ is considered. This corresponds to average channel capacities of $C_{avg}= 42.4 - 63.9$~Mbps.
%**** Here we need to give the bandwidth we assume and the resulting channel capacities. Also, we need to say the typical channel loads we consider.****
%We ran the simulations for a period of $10^{10}$ time slots, which allows empirical evaluation of the system performance up to a violation probability of $\eps =10^{-8}$.

%\todo[inline]{We concluded that the effect of $T_D$ is easy to see from the equations. We should add that as Corollary.}

\begin{figure}
\centering
\includegraphics [width=4in]{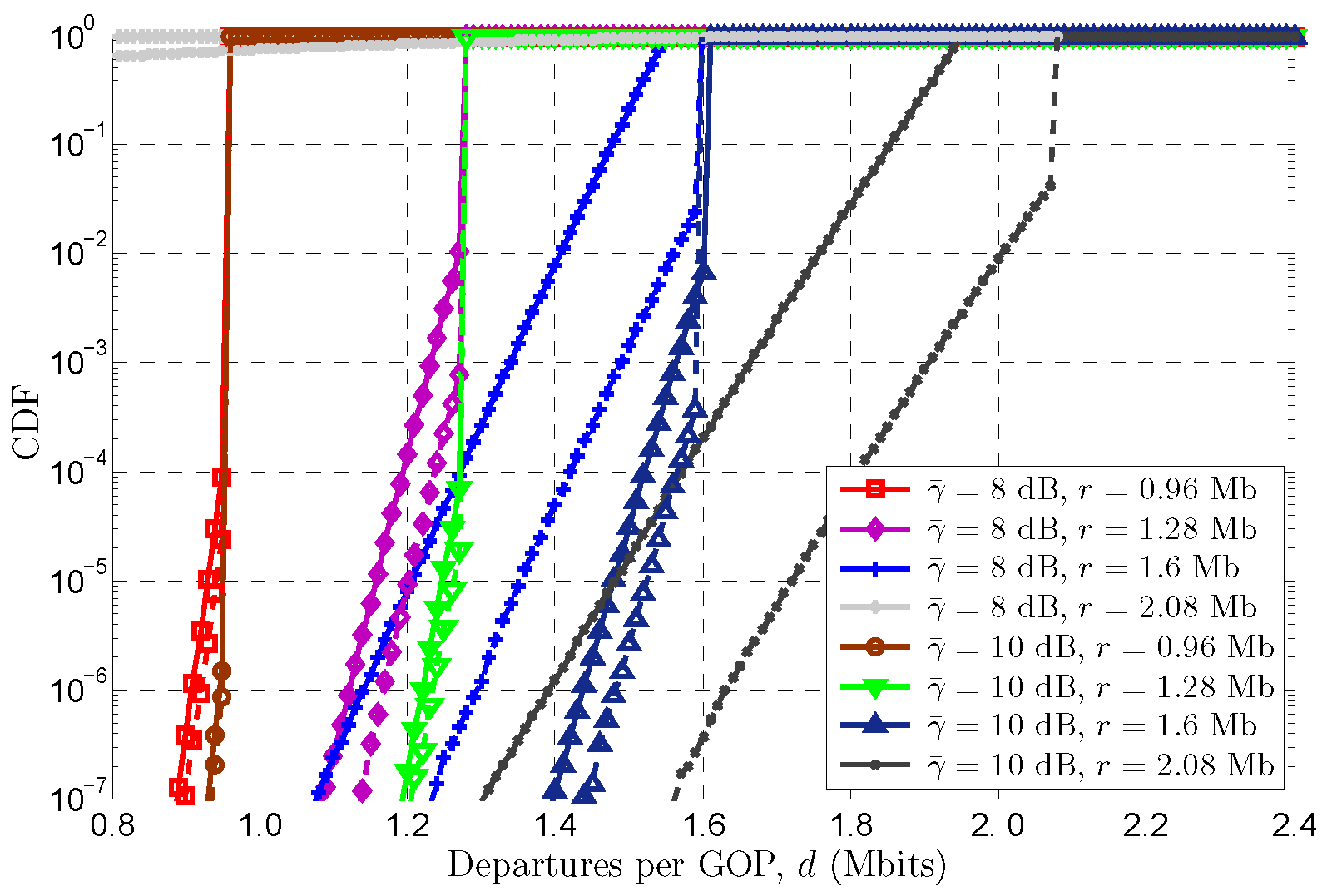}
%\vspace{-2mm}
\caption{Violation probability ($\eps^d$) (computed and simulated) vs. departure bound $d$ for SVC over multi-hop wireless network for different GOP size $r$ and for $\bar \gamma = 8, 10$ dB, with $T_D =450$ ms, $N=3$,   $W=2.2$~MHz    and    $n =2.5$ GOP/s.}
\label{fig:dep_CDF}
\end{figure}

Fig. \ref{fig:dep_CDF} shows the CDF of the per frame departures $d$, not yet considering the effect of the layering at the decoding, for $N=3$ and for various transmitted frame size  $r$ and average channel SNR $\bar{\gamma}$ values. For reference, the channel utilization for the case $r=2.08$ Mb is $0.8$ under $\bar{\gamma}= 10$ dB, and is $0.99$ for $\bar{\gamma}= 8$ dB.

The figure confirms that the model provides a lower bound on the number of bits received per frame, and shows that the empirical CDF shows the same exponential increase as the model-based lower bound. This exponential growth in $d$ can clearly be observed from   \eqref{eq:DBoundRayleigh11}.
The bound is tight for low and moderate load, but acceptable even for high utilization of $0.99$, specifically, the gradient for the model and simulation based results are equal which means that the  error diminishes as $\eps$ grows smaller.

We notice that reducing utilization, e.g., by reducing frame size for a given SNR, results in sharper curves, which means that the channel impairments have smaller effect on the video quality. On the other hand,  the figure shows that high utilization may lead to overload and low received quality, see for example the $0.99$ utilization case of $\bar{\gamma}=8$ and $r=2.08$, where the probability of receiving even $d=0.8$ is close to zero. These results reflect well that allowing transmission rates close to the average channel capacity would lead to overload and low quality streaming for latency critical applications.

\begin{figure}
\centering
\includegraphics [width=4in]{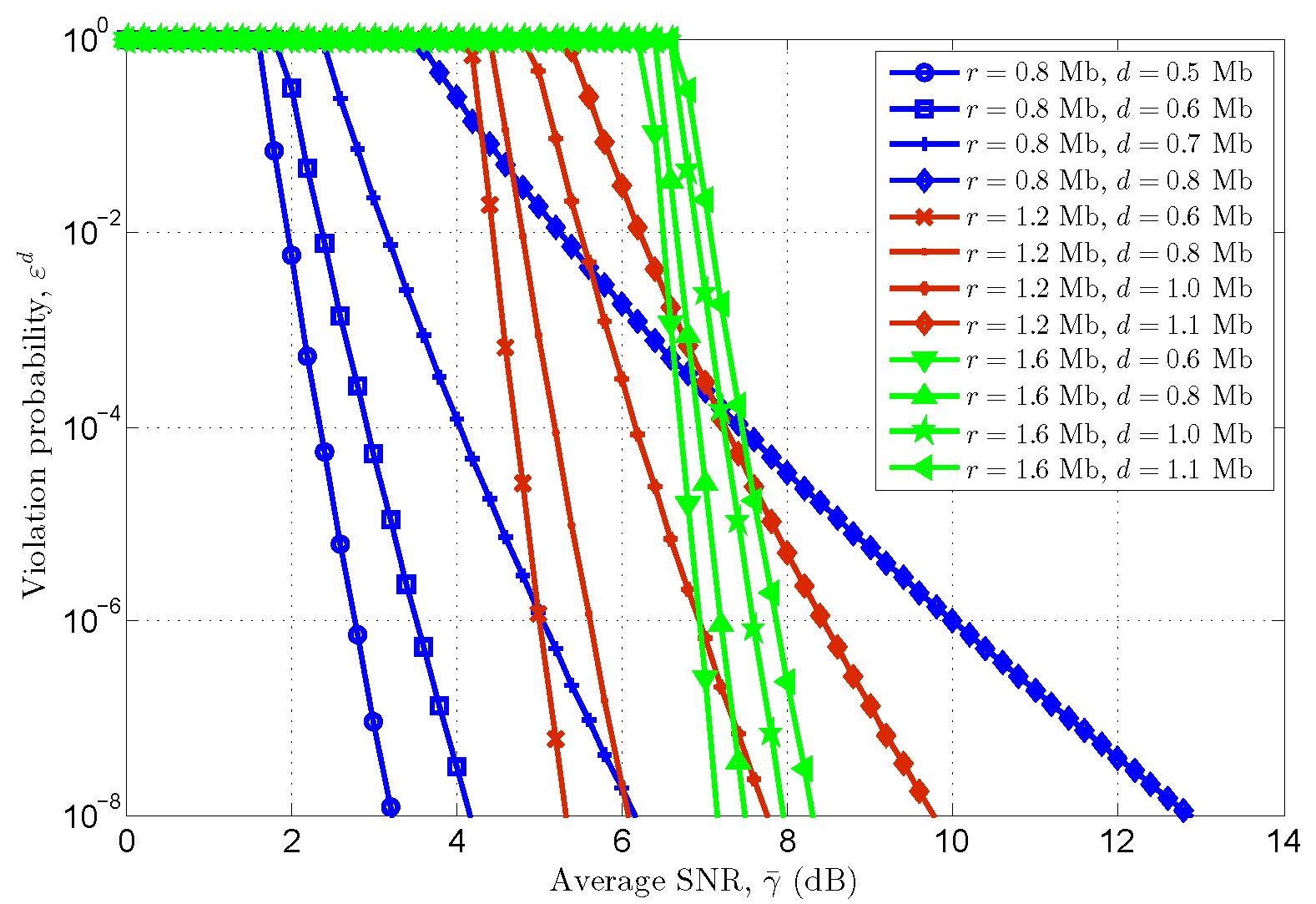}
%\vspace{-2mm}
\caption{Violation probability ($\eps^d$)  vs. average SNR ($\bar \gamma$) for SVC over multi-hop wireless network for three different GOP sizes $r= 0.8, 1.2$ and $1.6$ Mb  and for different departure within $T_D$ per frame $d$, with $T_D =450$ ms, $N=3$,   $W=2.2$~MHz    and    $n =2.5$ GOP/s.}
\label{fig:eps-vs-snr}
%\vspace{-2mm}
\end{figure}

\begin{figure}
\centering
\includegraphics [width=4in]{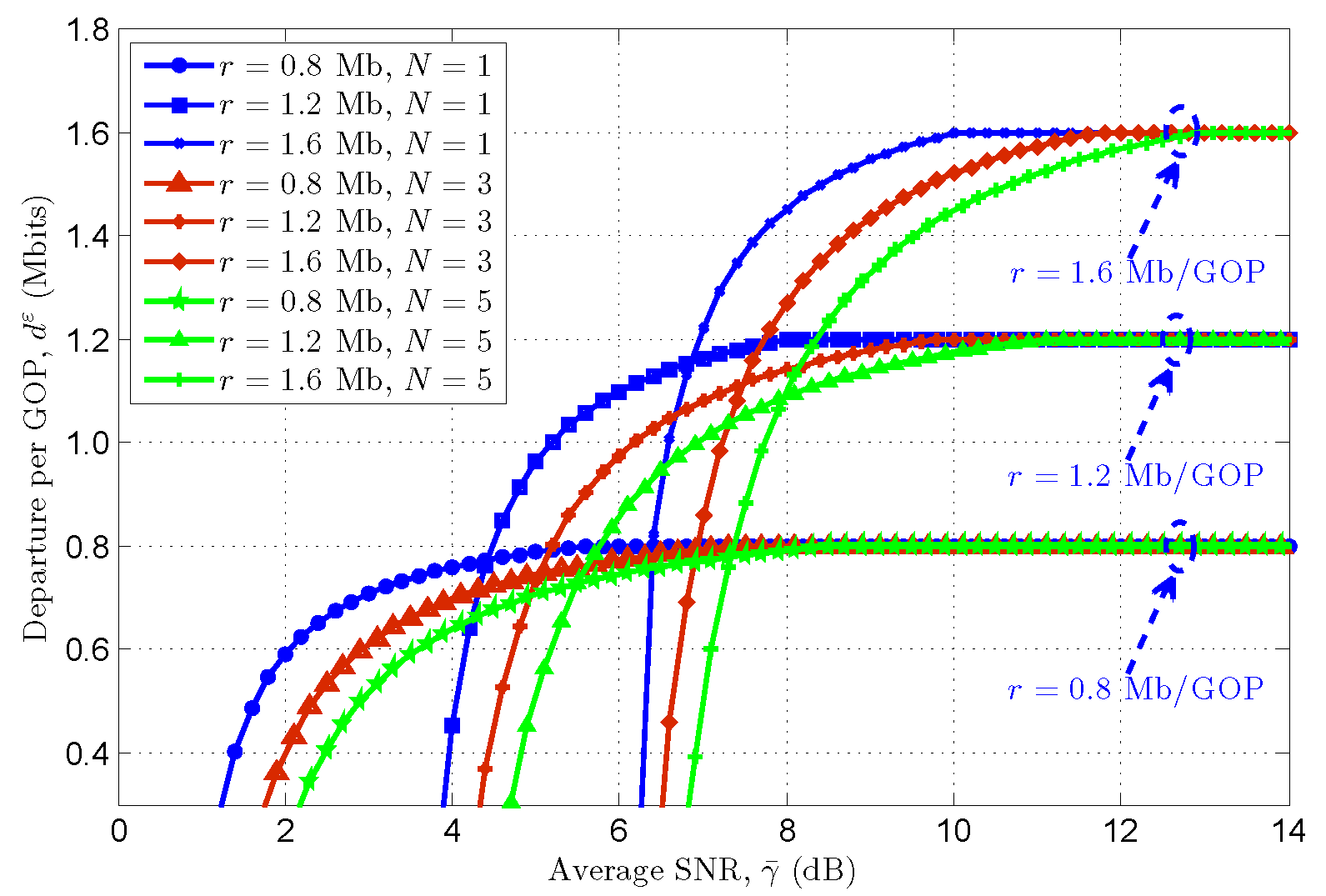}
%\vspace{-2mm}
\caption{Departure per frame ($d^{\eps}$)  vs. average SNR ($\bar \gamma$) for SVC over multi-hop wireless network for $\eps = 10^{-4}$, for three different GOP sizes $r= 0.8, 1.2$ and $1.6$ Mb  and for $N=1,3$ and $5$ hop, with $T_D =450$ ms,    $W=2.2$~MHz    and    $n =2.5$ GOP/s.}
\label{fig:dep-vs-snr}
%\vspace{-2mm}
\end{figure}

\begin{figure}[ht]
\centering
\includegraphics [width=4in]{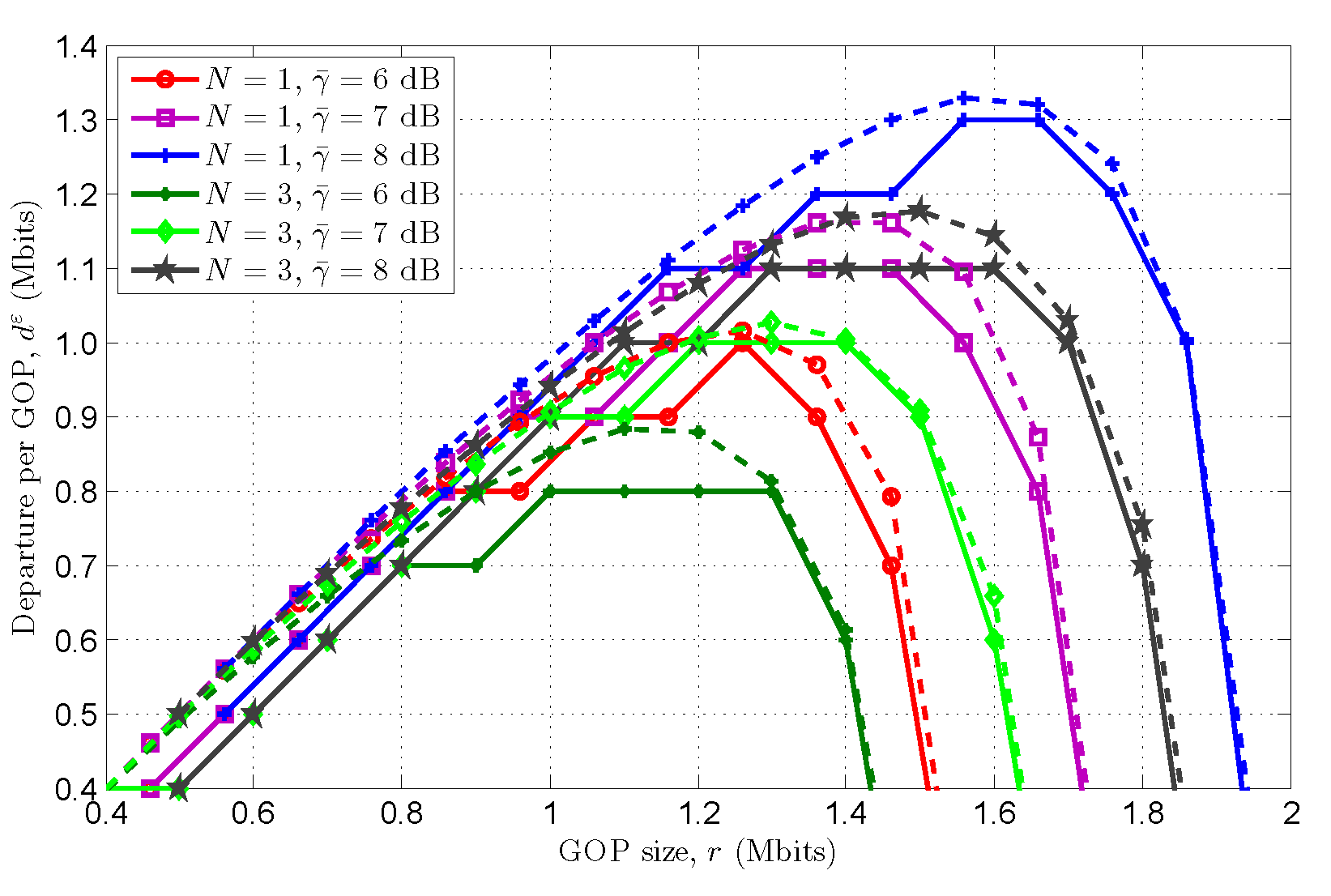}
%\vspace{-2mm}
\caption{Departure per frame ($d^{\eps}$)  vs. GOP size ($r$) for SVC over multi-hop wireless network (solid line for layered video frames and dashed line for fluid traffic model) with layer size $m=100$ kbits, for different  average SNR  ($\bar \gamma = 6,7,8$ dB) and for $N=1, 3$  hop, $\eps = 10^{-6}$, $T_D =450$ ms,     $W=2.2$~MHz    and    $n =2.5$ GOP/s.}
\label{fig:d-vs-r}
\end{figure}

Figs \ref{fig:eps-vs-snr} and \ref{fig:dep-vs-snr} evaluate the effect of the channel quality on the received video performance. Fig. \ref{fig:eps-vs-snr} shows for $N=3$ and for various $r$ and $d$, that the violation probability $\eps$ for given $d$ decreases almost exponentially when increasing the average SNR, as soon as the system becomes stable. Fig. \ref{fig:dep-vs-snr} shows how the per hop average SNR affects the per frame departures $d^{\eps}$ for a  violation probability $\eps=10^{-4}$, for different transmitted frame sizes $r$ and number of hops $N$. We can see that the SNR has significant effect on the optimal transmission scheme, for example, at $\bar{\gamma}=6$, $r=1.2$ Mbits provides the best performance among the considered frame sizes,  $r=0.8$ Mbits does not fully utilize the network, while $r=1.6$ Mbits leads to low quality due to network congestion. We also observe that the effect of number of hops, $N$, is significant at high utilization, but diminishes as $\bar{\gamma}$, and thus the channel capacity increases.
%This suggest a possible trade off of $N$ and $\bar \gamma$, and consequently that channel-aware, video quality based optimal  routing scheme for adaptive SVC may be designed based on the proposed model and the obtained probabilistic bounds.

\section{Adaptive video transmission and routing}\label{sec:Adaptive}
Our analysis exposes the effect of two extreme network behaviours that influence received video quality, namely, network congestion (at high  utilization) due to bad channel quality and/or high frame rate, and network underutilization due to low transmitted frame size. It also shows that the optimal operating point, where the transmitted frame size maximizes the received video quality depends on the channel conditions and on the length of the transmission path. Since the wireless channel quality  may vary with time, the optimal performance can be achieved by adapting the transmitted frame size to the SNR of the corresponding channels. It may also be beneficial to adapt the routing to the underlying channel quality. In this section, we examine both scenarios and provide examples to illustrate the benefits of such adaptation.

To evaluate the effect of under utilization as well as system overload, Fig.~\ref{fig:d-vs-r} shows the departures per frame, $d$, that fulfill the violation probability limit $\eps=10^{-6}$, as a function of the transmitted frame size $r$, for different SNR values and number of hops.

{ {The figure shows that the frame size leading to maximum departures per frame depends on both of the network parameters. Increasing the frame size above this maximizing value leads to fast quality degradation as the network becomes more saturated.}}
%The figure shows that there is an optimum transmitted frame size, and the optimal value depends on both of these network parameters. Increasing the frame size above the optimum value leads to fast quality degradation as the network becomes saturated.
In this figure we also show  the effect of layered transmission compared to its fluid counterpart, considering layer size of $m=100$ kbits. As layering affects both the possible transmitted and received frame sizes, we can see  performance degradation of a maximum of one layer size. Moreover, we can see that the same performance can be achieved under a range of transmitted frame sizes, which means that an adaptation algorithm would have to find the smallest value to maximize the performance under the lowest transmission rate and thus lowering the energy consumption.

Fig.~\ref{fig:eps-vs-r} compares the achieved violation probabilities as a function of the transmitted frame size $r$, for different per frame departure values $d$ and SNR values $\bar{\gamma}$, showing the analytic upper bounds as well as the simulation results. Again, we see that there is an optimum $r$ that minimizes $\eps$. This optimum depends significantly on $\bar{\gamma}$, and slightly also on the aimed received quality $d$. The simulation results reflect that even though the model overestimates the violation probability, the model-based optimization suggests $r$ values close to the real optimum, found via simulation.
%Note also, that the violation probability increases rapidly when moving away from the optimum frame size, which suggests that similar frame sizes will optimize the performance under the range of reasonable violation probability values. Consequently even if the model underestimates the reliability, the transmission rate it suggests is not far from the optimal transmission rate.
Consider for example $\bar{\gamma}=6$ and $\eps=10^{-6}$. The model predicts that $d=0.9$ Mbits can be achieved with the required reliability with $r=1.1$ Mbits, while according to the simulation results, the combination $d=1$ Mbits, $r=1.2$ Mbits is possible too. That is, the model-based parameter selection leads to $10$\% bitrate loss only, despite the slacknesss of the violation probability bound.

\begin{figure}
\centering
\includegraphics [width=4.in]{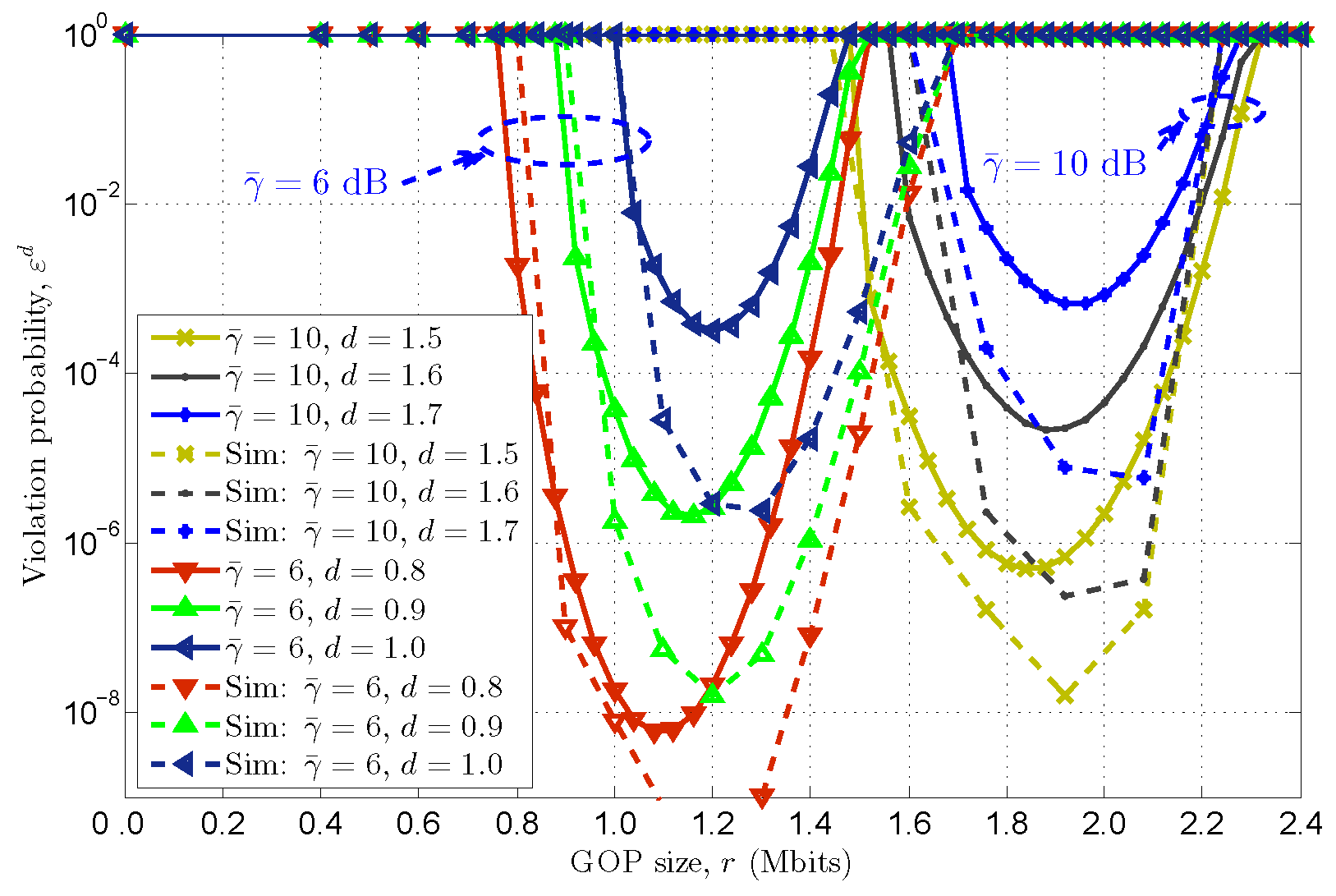}
%\vspace{-2mm}
\caption{Violation probability ($\eps^d$) (computed and simulated) vs. GOP size ($r$) for SVC over multi-hop wireless network (solid line for bounds and dashed line for simulated) for $\bar \gamma = 6, 10$ dB and for different target departure per GOP $d$, with $T_D =450$ ms, $N=3$,  $W=2.2$~MHz    and    $n =2.5$ GOP/s.}
\label{fig:eps-vs-r}
\end{figure}

\begin{figure}
\centering
\includegraphics [width=4in]{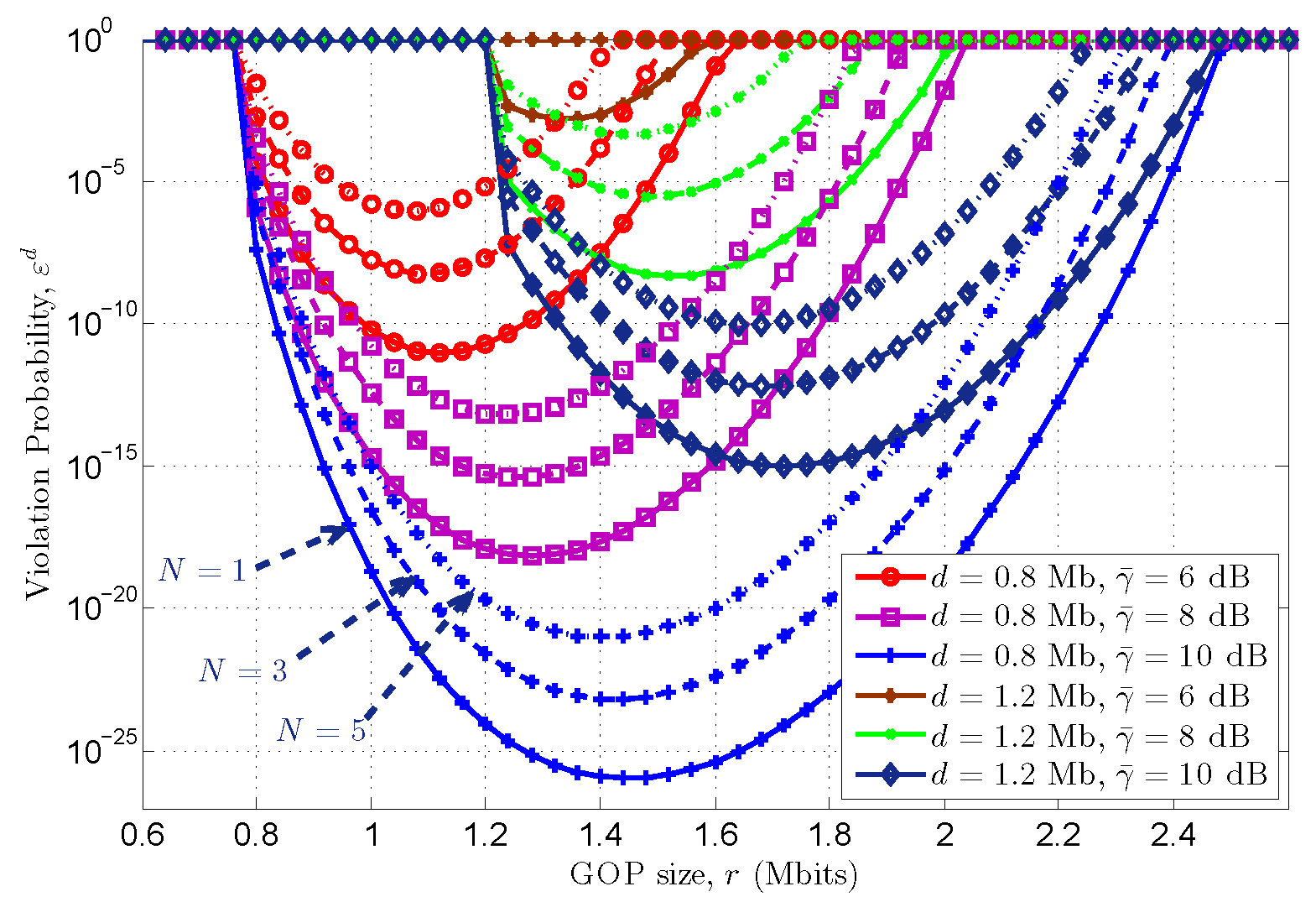}
%\vspace{-2mm}
\caption{Violation probability ($\eps^d$)  vs. GOP size ($r$) for SVC over multi-hop wireless network for $d=0.8, 1.2$ Mb and for different  average SNR per hop ($\bar \gamma =6,8,10$ dB) and $N=1,3, 5$ hop,  $T_D =450$ ms,    $W=2.2$~MHz    and    $n =2.5$ GOP/s.}
\label{fig:eps-vs-r_gamma}
\end{figure}

%{\color{red}
%Fig.~\ref{fig:eps-vs-r} also shows   rapid increase in the violation probability   when moving away from the optimum frame size. This suggests that the challenging application of reliable video streaming under low playout deadline benefits significantly from the availability of a large number of enhancement layers, which allows fine-grained rate adaptation according to the channel conditions and reliability constraints.
%}

Fig.~\ref{fig:eps-vs-r} also shows  a rapid increase in the violation probability   when moving away from the optimum frame size. Therefore, the availability of a large number of enhancement layers  is critical for a fine-grained rate adaptation to channel conditions, subject to reliability constraints. This becomes never more critical than  in applications that require reliable video streaming under low playout deadline, e.g., remote surgery, control of unmanned vehicles.

Fig. \ref{fig:eps-vs-r_gamma} summarizes the achievable performance for different expected received frame size values $d$,  SNR  and number of hops. We see that the range of transmitted frame sizes that yield acceptable violation probability depends  on $d$ on one side, and on $\bar{\gamma}$  on the other side. The optimal frame size is determined  by these two parameters, while the number of hops,   $N$, affects significantly the achievable violation probability, but not the optimal value of the frame size.

In order to  examine the efficiency of model-based frame size adaptation, we consider adaptation over a fixed transmission path and cross-layer optimized routing and rate adaptation. We compare the proposed model-based adaptation (MOD) to the optimal adaptation (OPT), where the optimum transmission frame sizes, and the resulting per frame departures are obtained by conducting extensive simulations. On Figures \ref{fig:adapt-N3} and \ref{fig:adapt-N1to3} we show the transmitted and received frame size $r$ and $d$ for OPT. For MOD we show the transmitted frame size that is suggested by the model, the bound on the received frame size, and the actual received frame size where the reliability constraint holds, derived through simulations.

%to an adaptation scheme that is  based on feedback from the decoder regarding the received video quality (QoS). Specifically, we consider an additive increase, multiplicative decrease (AIMD) scheme to quickly recover from congestion, where the number of transmitted layers, $L$ is halved in the case of quality decrease due to overload, and then increased by one layer at a time when no congestion is detected. We compare these solutions

%\todo[inline]{What is meant by stable state in ' we assume that the system is in a stable state'?}

\begin{figure}[ht]
\centering
\includegraphics [width=4in]{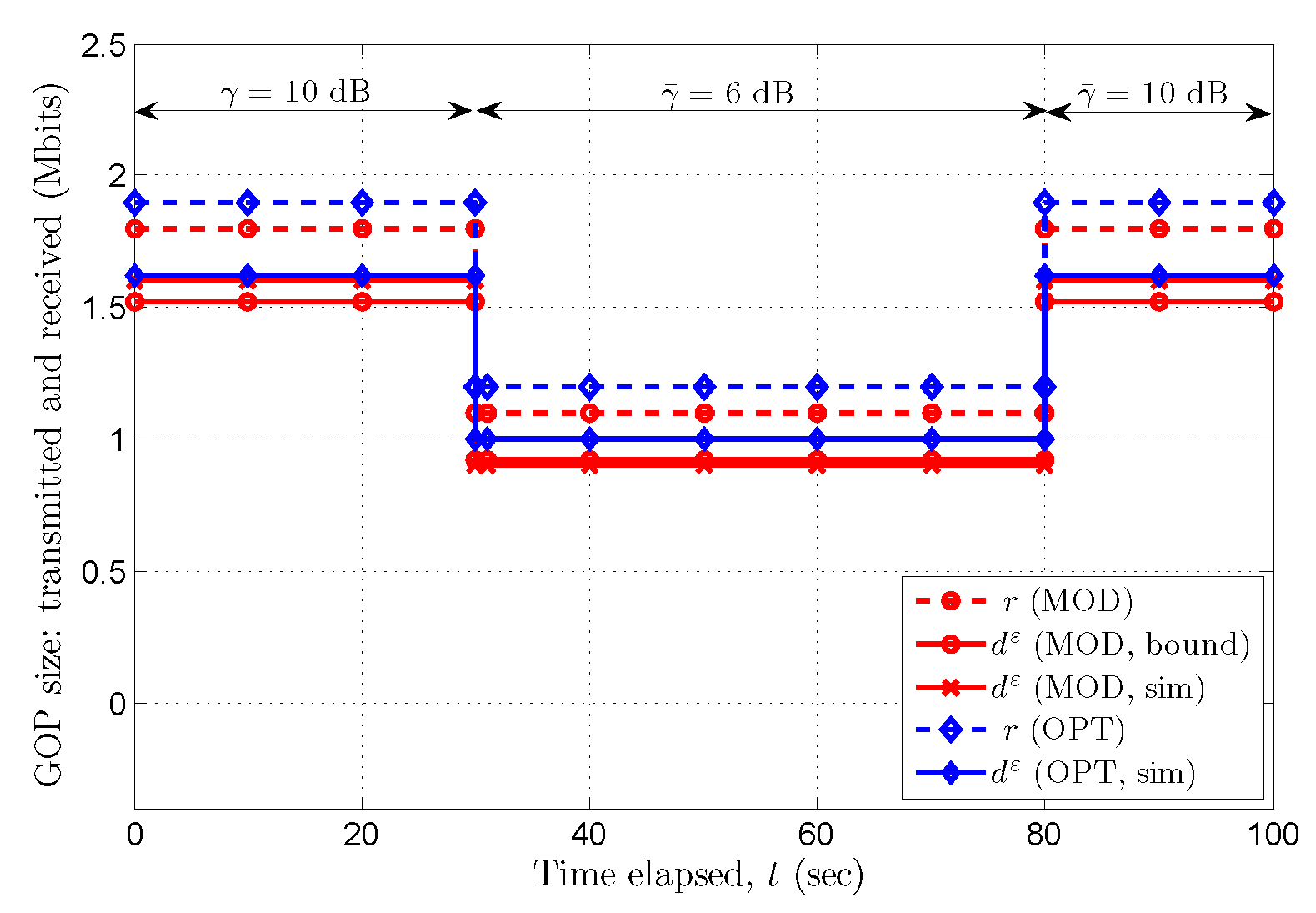}
%\vspace{-2mm}
\caption{Frame size ($r$) adaptation for SVC over 3-hop wireless network for the model-based adaptation (MOD) and for violation probability $\eps = 10^{-5}$ compared to the optimal adaptation (OPT)  when average SNR, $\bar \gamma =10 $ dB then it drops to $\bar \gamma =6 $ dB and get back to $\bar \gamma =10 $ dB again, for   $T_D =450$ ms,    $W=2.2$~MHz    and    $n =2.5$ GOP/s.}
\label{fig:adapt-N3}
%\vspace{-5mm}
\end{figure}

Fig. \ref{fig:adapt-N3} considers fixed routing with $N=3$, and layered coding with 100 kbits layer sizes. We consider a scenario where the SNR $\bar{\gamma}$ changes from 10 dB to 6 dB and back to 10 dB at times $t=30$ seconds and $t=80$ seconds respectively. We use results similar to the ones reported in Fig. \ref{fig:eps-vs-r} to demonstrate the frame size adaptation in time. We assume that both the OPT and the MOD based schemes have stabilized at $t=0$. OPT transmits with a frame size of $r=1.9$ Mbits, and receives a frame size of $d=1.6$ Mbits with violation probability $\eps=10^{-5}$. The model-based scheme slightly underestimates both $r$ and $d$, but due to the layering, it reaches the same actual per frame departures as the OPT solution. After the channel quality degradation, the MOD scheme decreases $r$, maintaining the system stability, again operating slightly below the OPT scheme. These results demonstrate that albeit the proposed network calculus based model provides only a lower bound on the per frame departures under some quality constraints, it enables the determination of  a near optimal transmission frame size as it was suggested by Fig.~\ref{fig:eps-vs-r}.

In a real implementation of the model-based scheme, the channel quality change would be followed by a transient phase, where the average SNR value is gradually updated, leading to a period with lower than optimal performance. The characterization of this transient phase is out of the scope of the paper.

%\todo[inline]{I think this last sentence should be emphasized somehow, may be by referring to the correspondin figure and commenting again on the small sensitivity around the optimum.}

\begin{figure}[ht]
\centering
\includegraphics [width=4in]{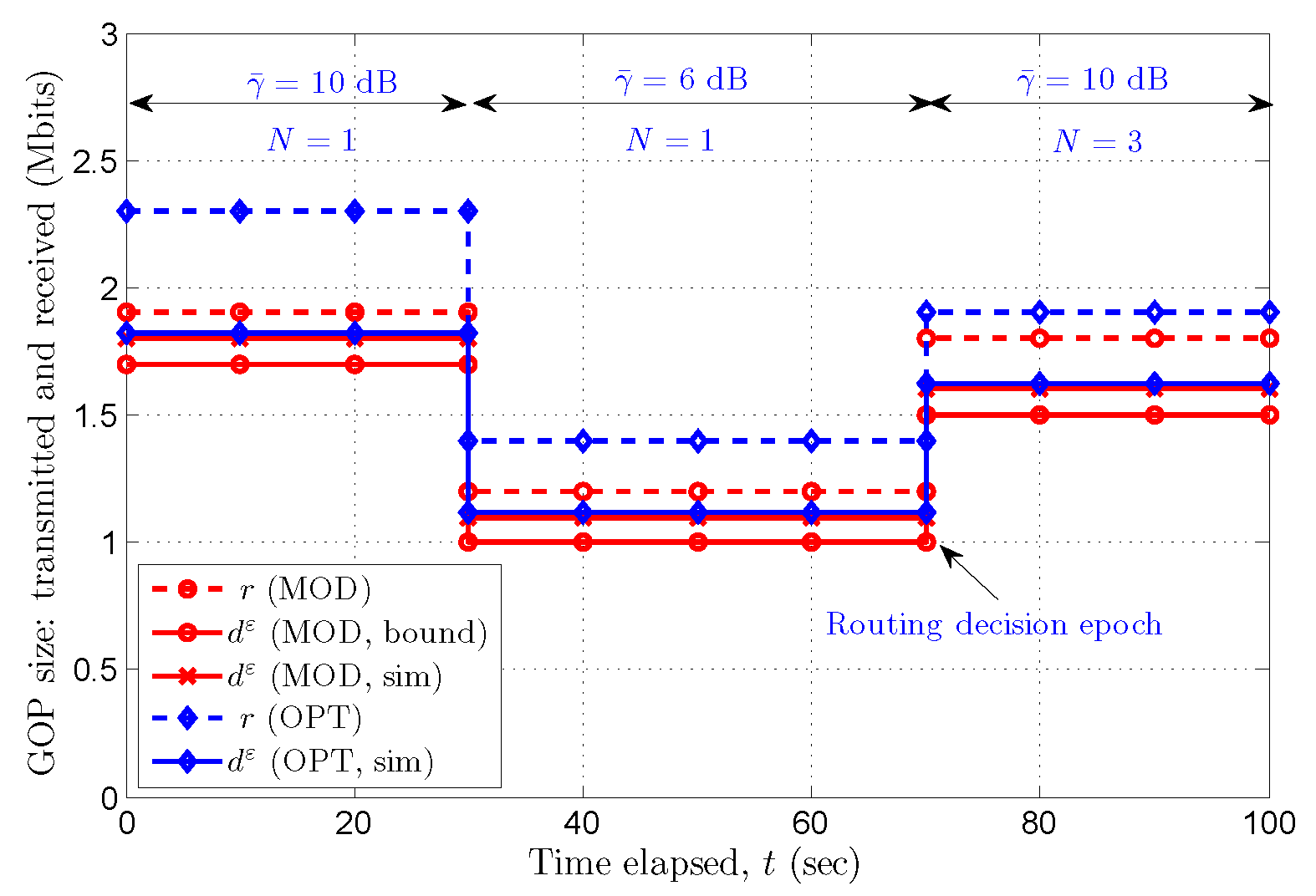}
%\vspace{-2mm}
\caption{Frame size ($r$) adaptation with routing for SVC over  wireless network for the model-based adaptation (MOD)  for violation probability $\eps = 10^{-5}$ compared to the optimal adaptation (OPT)  for a single hop link with  SNR $\bar \gamma =10 $ dB then it drops to $\bar \gamma =6 $ dB while a 3-hop link with $\bar \gamma =10 $ dB per hop exist, for $T_D =450$ ms    and    $n =2.5$ GOP/s.}
\label{fig:adapt-N1to3}
\end{figure}

Finally, Fig. \ref{fig:adapt-N1to3} demonstrates an example of rate adaptation combined with routing. We assume that the source node receives routing information, including the per link SNR values periodically, for example  every 30 seconds as suggested for the  RPL standard \cite{Accettura}. Between routing updates, the source performs rate adaptation based on the SNR feedback on the actual path. We consider the case when the quality of the single hop path deteriorates from $\bar{\gamma}=10$ dB to $\bar{\gamma}=6$  dB at $t=30$ seconds, and new routing information is received at $t=70$ seconds, about an $N=3$ path with 10 dB per link SNR. In this case, the longer path provides better service for the delay constrained transmission, as it has also been shown in Fig.~\ref{fig:d-vs-r}. As a result, the MOD scheme first adapts to the poor channel quality on the single hop path, it then selects the three-hop path, and increases the transmitted frame size according to the better channel conditions. As the reporting of per hop SNR values, or the minimum SNR perceived on a path can be easily accommodated in routing protocols like RPL, routing combined with the model-based rate adaptation provides an excellent approach to ensure reliable, high quality, delay sensitive video transmission in wireless networks.

\section{Conclusion}
\label{sec:conclusions}
In this paper we propose a network-calculus-based rate adaptation for delay-sensitive scalable video transmission over multi-hop wireless transmission paths. We derive new network calculus results that provide a probabilistic lower bound on the received video quality while considering the variability of the wireless channels, the  effect of the queuing delays at the network nodes and the frame-based playout at the receiver. Our evaluation shows that the channel quality has a more significant effect on the playout performance than the number of hops in the traversed path under low and moderate loads. Nonetheless, the effect of the hop count becomes significant as the network load increases. We show that even if the lower-bound-based model underestimates the achievable reliability, the transmission rate suggested by the model is close to the real optimum.
{  Our results also show that the performance degradation due to the layering effect, compared to the  perfect adaptation using the fluid model, depends significantly on the layer size, and hence, the number of enhancement layers per frame. That is, reliable, low latency video streaming over wireless links benefits greatly from adding more layers in layered coding. }

The proposed model provides a tool for low-complexity and fast adaptation of the number of transmitted layers  to  the underlying channel conditions, the playout delay limit and the desired reliability constraints. Our results show that the streaming performance under the model-based rate adaptation is very close to the achievable optimum for various network parameters (within 10\% in the considered numerical examples). This suggests that the proposed network-calculus-based approach is an efficient tool for channel-aware rate control and routing for adaptive layered video transmission under strict playout delay limits.
%{\color{green} Our results also show that the performance degradation due to the layering effect, compared to the  perfect adaptation using fluid model, depends on the layer size and hence the number of layers  per frame. }

%\vspace{2mm}

%\bibliographystyle{unsrt}
%\bibliographystyle{abbrv}
%\bibliography{IEEEabrv,../included-video-stream-references,..}{}
%
%\bibliographystyle{IEEEtran}
%\bibliography{included-video-stream-references}

\end{document}